\documentclass[aps,pra,10pt,twocolumn,superscriptaddress,amssymb,amsmath,longbibliography, showpacs,balancelastpage, nofootinbib]{revtex4-1}
\usepackage{CJKutf8}

\usepackage{amsthm}
\newtheorem{theorem}{Theorem}
\newtheorem{corollary}[theorem]{Corollary}
\newtheorem{lemma}[theorem]{Lemma}

\usepackage[utf8]{inputenc}
\usepackage{amsmath,amssymb,amsfonts,stmaryrd}
\usepackage{physics}
\usepackage{dsfont}
\usepackage{graphicx}
\usepackage{xcolor}
\usepackage{graphicx}
\usepackage{cancel}
\usepackage{natbib}
\usepackage{color}
\usepackage{tikz}
\usepackage{mathrsfs}
\usepackage{relsize}
\usepackage{multirow}
\usepackage[linktocpage]{hyperref}
\usepackage[all]{xy}
\usepackage[normalem]{ulem}
\hypersetup{colorlinks=true,citecolor=blue,linkcolor=blue, urlcolor=blue, breaklinks=true}

\usepackage{float}

\usepackage{bm} 
\usepackage{subfigure} 
\usepackage{environ}
\usepackage{url}
\usepackage[margin=0.75in]{geometry}
\usepackage{tikz}
\usepackage{tikz-cd}
\usepackage{adjustbox}

\usepackage{mathtools}

\newcommand{\Rep}{\mathbf{Rep}}

\newcommand{\Sq}{{\mathrm{Sq}}}

\newcommand{\Z}{{\mathbb Z}}

\def\cZ{\mathcal Z}
\def\cC{\mathcal C}

\newcommand{\fB}{{\mathfrak B}}
\newcommand{\fC}{{\mathfrak C}}
\newcommand{\fD}{{\mathfrak D}}
\DeclareMathOperator{\End}{End}
\newcommand{\tVect}{\mathbf{2Vect}}
\newcommand{\SVect}{\mathbf{SVect}}

\newcommand{\tSVect}{\mathbf{2SVect}}
\newcommand{\tRep}{\mathbf{2Rep}}
\newcommand{\Mod}{\mathbf{Mod}}
\newcommand{\sPic}{\mathscr{P}ic}
\newcommand{\sBrPic}{\mathscr{B}r\mathscr{P}ic}

\newcommand{\sAut}{\mathscr{A}ut}
\newcommand{\rH}{\mathrm{H}}
\newcommand{\SH}{\mathrm{SH}}

\newcommand{\br}{\mathrm{br}}

\newtheorem*{claim*}{Claim}

\theoremstyle{definition}
\newtheorem{defn}{Definition}
\newtheorem{ansatz}[defn]{Ansatz}
\newtheorem{prop}[theorem]{Proposition}
\theoremstyle{definition}
\newtheorem*{notation*}{Notation}
\newtheorem*{trick*}{Trick}

\newtheorem*{convention*}{Convention}

\newcounter{mainthm}

\newtheorem{ex}[defn]{Example}

\def\Sq{\mathop{\mathrm{Sq}}\nolimits}

\NewEnviron{eqs}{%
\begin{equation}\begin{split}
    \BODY
\end{split}\end{equation}}

\definecolor{dkgreen}{rgb}{0,0.5,0}

\theoremstyle{definition}

\theoremstyle{remark}
\newtheorem*{remark}{Remark}

\makeatletter
\def\l@subsection#1#2{}    
\def\l@subsubsection#1#2{} 
\makeatother

\makeatletter
\def\@fnsymbol#1{\ensuremath{\ifcase#1\or \bullet\or \downarrow\or \Downarrow\else\@ctrerr\fi}}
\makeatother

\begin{document}

\begin{CJK*}{UTF8}{bsmi}

\title{Realizing Braided Fusion 2-Categories as (3+1)D Mixed-State Topological Orders}

\author{Ryohei Kobayashi}
\email[E-mail: ]{ryohei.k@ap.t.u-tokyo.ac.jp}
\affiliation{Department of Applied Physics, University of Tokyo, Bunkyo-ku, Tokyo 113-8656, Japan}

\author{Abhinav Prem}
\email[E-mail: ]{aprem@bard.edu}
\affiliation{Physics Program, Bard College, 30 Campus Road, Annandale-on-Hudson, New York 12504, USA}

\author{Matthew Yu}
\email[E-mail: ]{yumatthew70@gmail.com}
\affiliation{School of Natural Sciences, Institute for Advanced Study, 1 Einstein Drive, Princeton, NJ 08540, USA}

\date{\today}
\begin{abstract}
We construct large families of mixed-state topological orders in (3+1)D, including intrinsically mixed-state topological orders. We realize these theories by relating the data parameterizing braided fusion 2-categories to the strong higher-form symmetries and their anomalies that characterize mixed-state topological order. We focus here on bosonic theories, which include theories which admit an emergent fermion in the spectrum of line operators. Our general framework is exemplified on the lattice, where we present concrete models and identify how the data of braided fusion 2-categories are physically encoded. We then formalize this data and show that every braided fusion 2-category arises from a general construction (decoherence followed by strong and weak gauging) as the category of topological operators of a (3+1)D mixed state, with intrinsically mixed-state topological orders distinguished by the presence of strong higher-form symmetries whose anomalies cannot be realized in any (3+1)D gapped ground state. Notably, we show that (3+1)D admits a new mechanism with no (2+1)D counterpart: intrinsically mixed-state order whose topological operators braid non-degenerately, but which carries a gravitational anomaly.
\end{abstract}

\maketitle

\end{CJK*}

\tableofcontents



\section{Introduction}
\label{sec:intro}

Topological states of matter provide canonical resources for fault-tolerant
quantum information storage and processing. Their recent preparation on programmable quantum simulators~\cite{satzinger2021,semeghini2021,iqbal2023,xu2023,andersen2023,fossfeig2023,fib1} has rendered the fate of topological pure states under noise a question of immediate practical importance. It is becoming increasingly clear, however, that decoherence does not merely degrade patterns of entanglement present in pure initial states: rather, the resulting density matrices can host new forms of order without pure-state counterpart. Indeed, the distinction between ``strong" and ``weak" symmetries for mixed states~\cite{buca2012,albert2014sym} opens the door to strong-to-weak spontaneous symmetry breaking (SWSSB)~\cite{lee2023,sala2024,lessa2025sw,kuno2025sw,zhang2025} and intrinsically average symmetry protected topological (iASPT) mixed-state phases~\cite{deGroot2022og,lee2023aspt,zhang2022strange,ma2023prx,ma2024,xue2024,guo2024,Ma2025,luo2025top,zhang2025sym,qi2025taco}.

From the perspective of quantum error correction, a fundamental question is understanding how logical information stored in \textit{topologically ordered} pure states gets corrupted by noise. For the two dimensional toric code, for instance, the encoded information remains recoverable for noise below a certain threshold, beyond which it is no longer accessible, signaling a phase transition in the decodable information encoded in the decohered mixed state~\cite{dennis2002,bao2023mixed}. Such transitions mark the destruction of the parent topological order; more surprisingly, decoherence can also generate new forms of topological order. Indeed, recent work has shown that subjecting (2+1)D topologically ordered states to local noise can produce \textit{intrinsically mixed-state topological order} (imTO), which is a form of topological order that cannot arise in the ground states of any 2D gapped local Hamiltonian~\cite{sohal2024imto,ellison2024,Wang2025intrinsic}. Such mixed states are characterized in terms of their strong one-form symmetries and their corresponding self/mixed anomalies, and admit a partial classification in terms of non-modular braided fusion categories~\cite{sohal2024imto,ellison2024,Hsin:2025pqb}.

Here, the anomalies of strong higher-form symmetries\footnote{A mixed state satisfying $U_g \rho U_g^\dagger = \rho \, \forall g \in G$ is weakly symmetric under $G$ while one that satisfies $U_g \rho = e^{i \theta_g} \rho$ for $g \in G$ is said to be strongly symmetric. The former describes a stochastic ensemble of states where the state is symmetric on average but individual states need not be symmetric, whereas the latter describes the case where each state within the ensemble is independently symmetric under $G$.} play a central role: they underpin the quantum information retained by a mixed-state topological order, and they guarantee its stability as a genuine mixed-state phase of matter (see also \cite{Feng2026higher} for recent discussions on lattice characterization of higher-form anomalies). Following Ref.~\cite{sang2023mixed}, we regard two mixed states as belonging to the same phase when they are related by (distinct) finite-depth local quantum channels acting in \emph{both} directions. Since a state invariant under anomalous strong higher-form symmetries cannot be generated from a trivial state via a finite-depth local quantum channel, these symmetries and their anomalies furnish robust data characterizing mixed-state phases~\cite{li2024anyon,Hsin:2025pqb,lessa2025higher}. In fact, anomalous strong higher-form symmetries enforce intrinsic long-range entanglement: any state invariant under such anomalous symmetries has exponentially small overlap with every short-range entangled state~\cite{Hsin:2025pqb}.

While mixed-state topological order remains far from well-understood in general, one systematic approach is to subject pure-state topological orders (arising in ground states of gapped local Hamiltonians) to decoherence channels and analyze the symmetries and their anomalies of the resulting mixed states. We can then distinguish between two kinds of mixed-state topological orders. In the first case, the mixed state is only a ``weakly" deformed version of the parent pure state and the strong symmetries and their anomalies remain intact. The mixed state is then two-way channel connected to the parent pure state and, in the current taxonomy, is referred to as simply carrying mixed-state topological order (mTO). In contrast, the mixed state could result from subjecting a gapped ground state to strong decoherence such that the resulting mixed state cannot be two-way channel connected to the parent pure state. For instance, some of the initial strong symmetries can be broken to weak symmetries, with the resulting algebra of remaining strong symmetry anomalies impossible in any gapped ground state. Mixed states carrying such strong higher-form symmetry anomalies are said to realize imTO. This distinction can be re-stated physically in terms of \emph{remote detectability}: incoherently proliferating an excitation can leave behind operators which braid trivially with all others. When such ``transparent" operators are themselves anomalous (e.g., carry fermionic self-statistics), they obstruct a pure-state realization outright, since the anomaly prevents two-way channel connectivity to a trivial state and since remote detectability (or braiding nondegeneracy) is expected to hold in the ground state of any gapped local Hamiltonian in any dimension. In (2+1)D, this breakdown is precisely the failure of the anyon theory to be modular~\cite{sohal2024imto,ellison2024}. Neither our definition of a mixed-state phase nor the principle whose violation defines imTO makes reference to spatial dimension, however, and both accordingly carry over unchanged to higher dimensions.

In this paper, we propose a general framework for characterizing mixed-state topological order in (3+1)D, and use it to construct imTOs whose excitations are either all bosons or contain an emergent fermion in the spectrum of line operators. Our starting point is the ansatz that the
topological operators of a (3+1)D mTO form a braided fusion 2-category which need not be non-degenerate. This allows for the presence of transparent operators i.e., operators with trivial mutual braiding with every operator in the theory. Such an operator need not be trivial on its own: it can still carry nontrivial self-statistics, as with a fermionic particle or a fermionic loop. Such anomalous transparent operators violate the principle of remote detectability, prevent two-way channel connectivity to a pure-state TO, and enforce long-range entanglement. This ansatz converts a question about density matrices into one about categorical data: mathematically, mTO in (3+1)D (two-way channel connected to a pure state in our taxonomy) corresponds to a non-degenerate braided fusion 2-category $\mathcal{D}$, where a faithful symmetry grading can be forgetfully removed to recover pure-state topological data. imTO in (3+1)D arises in two distinct ways. Generalizing the (2+1)D construction, it is represented by a \emph{degenerate} braided fusion 2-category $\mathfrak{B}$ whose transparent operators generate anomalous strong symmetries; these obstruct two-way channel connectivity to pure-state topological order and enforce intrinsic long-range entanglement~\cite{li2024anyon,Hsin:2025pqb,lessa2025higher}. Alternatively, $\mathfrak{B}$ may be non-degenerate but carry a gravitational anomaly that no strictly (3+1)D ground state can realize; this provides a novel route to imTO that has no counterpart in (2+1)D. Note that degeneracy alone is insufficient: transparent operators that are anomaly-free bosons arise from essentially classical ensembles, which are two-way connected to trivial states.

In (3+1)D, topological excitations are both point-like and loop-like, such that topological operators organize into a braided fusion 2-category rather than a braided fusion category, opening more mechanisms for generating imTO via decoherence than are available in (2+1)D. For instance, the decohered object in (3+1)D can be a loop and, since loops in a purely bosonic theory can carry nontrivial self-statistics (and are thus dubbed fermionic loops)~\cite{Thorngren:2014pza, FHH21, Chen2023gravitational, Kobayashi2026statistics}, decohering the bosonic $\mathbb{Z}_2$ toric code already yields a nontrivial imTO~\cite{Hsin:2025pqb}. Most strikingly, we construct an imTO of a distinct kind: it hosts a fermionic particle and a fermionic loop with nontrivial mutual braiding, so that every operator is detected by braiding and remote detectability holds\footnote{In (3+1)D, line operators cannot braid with one another, so they always form a symmetric fusion category which may contain an emergent fermion; here, non-degeneracy instead concerns how lines are detected by surfaces. In (2+1)D a non-degenerate braided fusion category is never symmetric unless trivial.}. It realizes the (3+1)D $\mathbb{Z}_2$ gauge theory with a global gravitational anomaly~\cite{FHH21, Chen2023gravitational}, which is believed not to admit a pure-state realization in strictly three spatial dimensions without coupling to a four-dimensional bulk. This theory is thus intrinsically mixed not because of transparent operators, but because of its gravitational anomaly. Note that since fusion 2-categories with no nontrivial line operators are necessarily grouplike~\cite{JFY}, every theory we encounter can be obtained by gauging symmetries of a parent gauge theory and can be labelled by group-cohomological data. This is in contrast to (2+1)D, where modular tensor categories include non-group theoretic examples such as the Fibonacci anyon theory.

We exemplify our construction on the lattice, starting with a $\mathbb{Z}_N$ toric code and applying an error channel that incoherently proliferates loops carrying an 't~Hooft anomaly, resulting in a degenerate theory with transparent operators. We then gauge the 0-form symmetry $G$ of the degenerate theory, such that the non-anomalous part $N\subset G$ of the gauged 0-form symmetry acts strongly while the quotient $H=G/N$ acts weakly. Note that even if the left $H$ action on the density matrix is anomalous, one can still gauge $H$ as a weak symmetry; gauging a weak symmetry is also called ``classical gauging," which amounts to taking a classical ensemble of gauge field configurations~\cite{ellison2024}. Categorically, we show that the data obtained by this procedure is exactly the data parametrizing braided fusion 2-categories as presented in Theorem \ref{thm:mainA}: a finite abelian group $E$ together with a class $\beta$ specifying the proliferated loops, a faithfully graded crossed extension encoding the strong symmetry $N$, and a group $H=G/N$ acting by braided autoequivalences encoding the weak one.
In this language, pure-state (3+1)D TOs arise from \emph{faithfully} graded $G$-crossed braided fusion 2-categories, for which equivariantization reinstates the line operators and preserves nondegeneracy; mixed-state TOs unlock symmetry structures with non-faithful grading, enabled by gauging weak symmetries with anomalous left/right action.

We stress that this does not amount to a classification of (3+1)D mixed-state phases under the two-way channel equivalence relation (or the finer classification under the finite-Markov length condition~\cite{sang2024}). Our results rest on the Ansatz~\ref{ansatz:mTO}, which we motivate but do not prove; we restrict throughout to \emph{bosonic} theories (but allow for emergent fermions); and, we do not establish that distinct categorical data necessarily label states inequivalent under two-way channels, nor do we capture chiral mixed-state phases or multi-loop braidings. Our main result is to establish that \textit{every} braided fusion 2-category in the two families above is realized via a single construction -- incoherently proliferating possibly anomalous excitations of an Abelian gauge theory, followed by strong and weak 0-form gauging -- providing a physical realization for this categorical data. We implement this construction explicitly on the lattice for representative examples, including a new imTO with no (2+1)D counterpart.

This paper is organized as follows: In Sec.~\ref{sec:latmod}, we introduce (3+1)D imTOs from the point of view of lattice models. In particular, we cover the constructions of some bosonic theories that contain only bosonic point-like excitations as well as emergent fermions. In Sec.~\ref{section:generalconstruction}, we abstract from these examples a general construction. In Sec.~\ref{section:allbraided}, we study (3+1)D mTOs as braided fusion 2-categories, obtain the data parametrizing all such theories, and show that it coincides with the data produced by the construction of Sec.~\ref{section:generalconstruction}. We collect some relevant categorical facts about fusion 2-categories in Appendix~\ref{appendix:F2C}.

\section{Lattice models}
\label{sec:latmod}

We begin by example, describing various lattice constructions for (i)mTO in (3+1)D. The examples here serve as the starting point from which we will develop the general prescription for encoding every braided fusion 2-category as the (i)mTO of some (3+1)D mixed state. The key new example contained in this Section is an imTO that hosts an emergent fermion and an emergent fermionic loop, following which we present the general construction in Sec.~\ref{section:generalconstruction}.


\subsection{Intrinsically mixed-state topological order with a single fermion}
\label{subsec:singlefermion}

\begin{figure}[b]
\centering
\includegraphics[width=0.45\textwidth]{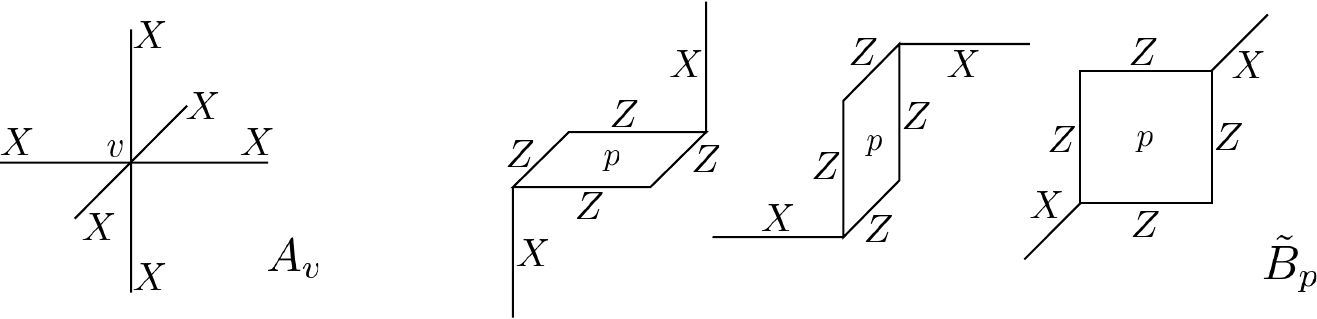}
\caption{Vertex and plaquette terms for the Hamiltonian defining the (3+1)D fermionic $\Z_2$ toric code.}
\label{fig:fTC}
\end{figure}

Our first example is an imTO that can be obtained by subjecting the (3+1)D $\Z_2$ \textit{fermionic} toric code~\cite{levinwen2003,vk2015ww,chen2019bos} to a local noise channel. We consider a 3d cubic lattice with a single qubit on each edge. We start with the following Hamiltonian for the (3+1)D fermionic $\Z_2$ toric code
\begin{align}
    H_{\text{fTC}} = -\sum_v A_v - \sum_p \tilde B_p ~,
    \label{eq:fTC}
\end{align}
where the vertex terms $A_v$ and the plaquette terms $\tilde B_p$ are shown in Fig.~\ref{fig:fTC}. This model realizes the deconfined phase of (3+1)D $\mathbb{Z}_2$ gauge theory that contains fermionic charge and bosonic loop excitations, and is hence a bosonic theory with an emergent fermion in the spectrum of topological operators.

This model has an anomalous $\Z_2$ 2-form symmetry generated by a line operator
\begin{align}
    L_{\text{f}}(\gamma) = \prod_{e\subset\gamma} L_{e}~,
\end{align}
where $\gamma$ is a closed loop of the cubic lattice, and the operator $L_e$ is shown in Fig.~\ref{fig:Le}. This strong $\Z_2$ symmetry has an 't Hooft anomaly characterized by a (4+1)D response
\begin{equation}
\label{eq:resp1}
\pi\int Sq^2 B=    \pi\int w_2\cup B~,
\end{equation}
where $B$ is a 3-form $\Z_2$ background. This 't Hooft anomaly corresponds to the fermionic statistics of the topological line operator that generates the symmetry.

We now consider an error channel which preserves the above anomalous $\Z_2$ symmetry, and leads to the incoherent proliferation of fermionic point-like particles:
\begin{align}
    \mathcal{N} = \prod_e\mathcal{N}_e, \quad \mathcal{N}_e(\rho) = p\rho + (1-p) \tilde L_e \rho \tilde L_e^\dagger~,
\end{align}
where the operators $\tilde L_e$ are shown in Fig.~\ref{fig:Le}. One can check that $L_e$ and $\tilde L_{e'}$ are commutative for all edges $e,e'$; therefore, the state obtained via this decoherence channel preserves the strong $\Z_2$ 2-form symmetry generated by $L_{\text{f}}(\gamma)$. 

\begin{figure}[b]
\centering
\includegraphics[width=0.4\textwidth]{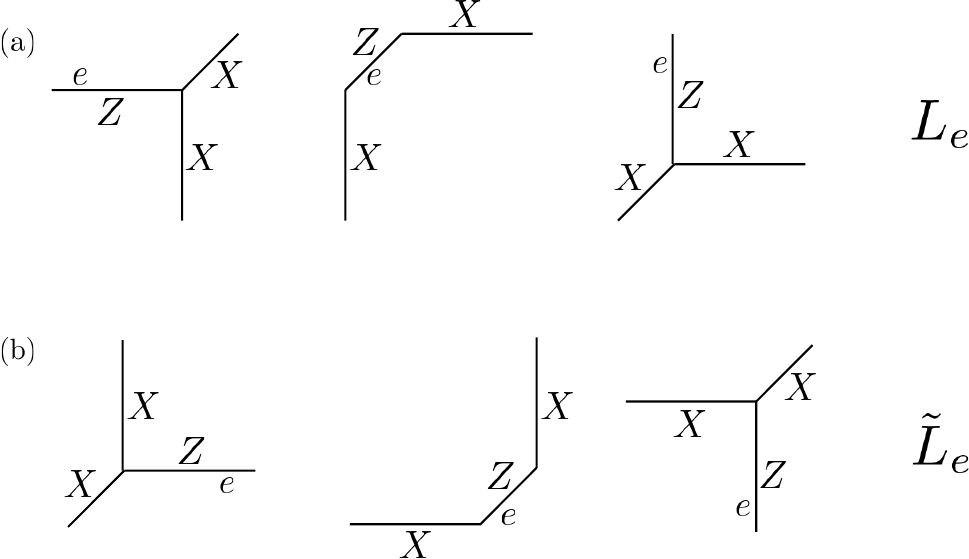}
\caption{Definitions of the operators $L_e, \tilde L_e$.}
\label{fig:Le}
\end{figure}

Here, the choice of $\tilde{L}_e$, rather than $L_e$, as the Kraus operator is essential. While both operators create a pair of fermions at the endpoints of $e$, they correspond to different framings of the fermion string. The $L_e$ do not mutually commute, reflecting the fermionic exchange statistics of the particles they transport, and a closed string $L_{\mathrm{f}}(\gamma)$ does not commute with $L_e$ on edges adjacent to $\gamma$; a channel built from $L_e$ would therefore not preserve $L_{\mathrm{f}}(\gamma)$ as a strong symmetry. In contrast, $\tilde{L}_e$ commutes with $L_{e'}$ for all $e,e'$, so that $\mathcal{N}$ commutes with every closed fermion string (and with every open string segment). Hence, for all $p$,
\begin{equation}
    L_{\mathrm{f}}(\gamma)\, \mathcal{N}\rho_0 = \lambda_\gamma\, \mathcal{N}\rho_0 ,
\end{equation}
where $\lambda_\gamma = \pm 1$ is the eigenvalue of $L_{\mathrm{f}}(\gamma)$ in the parent ground state $\rho_0 = |\mathrm{fTC}\rangle\langle \mathrm{fTC}|$.

It is convenient to represent the density matrix $\rho =\sum_{j,k}\rho_{j,k}\ket{j}\bra{k}$ as a pure state in a doubled Hilbert space~\cite{choi1975,jamio1972} in terms of the Choi state $\ket{\rho}\rangle = \sum_{j,k}\ket{j}\ket{k}^* \in \mathcal{H}_+\otimes \mathcal{H}_-$, where $\ket{k}^*:= K\ket{k}$ ($K$ denotes complex conjugation) and $\mathcal{H}_+ (\mathcal{H}_-)$ refer to the ket (bra) Hilbert spaces respectively. At the maximally decohered point $p=1/2$, the Choi state of the decohered density matrix is
\begin{equation}
    |\mathcal{N}\rho_0\rangle\!\rangle
    = \prod_e \left( \frac{1 + \tilde{L}_e^{+} (\tilde{L}_e^{-})^{*}}{2} \right)
      |\mathrm{fTC}\rangle_{+} |\mathrm{fTC}\rangle_{-} ,
\end{equation}
which is a zero correlation length state. Since $\mathcal{N}$ is a Pauli channel and $\rho_0$ is a stabilizer state, the $p=1/2$ channel simply twirls $\rho_0$ over the group generated by $\{\tilde{L}_e\}$: every element of $\langle A_v, \tilde{B}_p \rangle$ that anticommutes with some $\tilde{L}_e$ is projected out, while those that commute with all $\tilde{L}_e$ survive. Because $\tilde{L}_e$ anticommutes with $A_v$ at both endpoints of $e$, every product of vertex terms is eliminated. The surviving stabilizers are generated by dressed plaquette operators $\hat{B}_p \equiv \tilde{B}_p A_{v(p)}$, where $v(p)$ is a corner of $p$ fixed by the framing of $\tilde{L}_e$\footnote{With the conventions of Figs.~1 and 2, and $p$ the plaquette with lowest corner $r$: $v(p) = r+\hat{x}+\hat{y}$, $r$, and $r+\hat{x}$ for $xy$, $yz$, and $xz$ plaquettes, respectively.}. The maximally decohered state is thus
\begin{equation}
    \mathcal{N}\rho_0 |_{p=1/2}
    \propto \prod_p \frac{1 + \hat{B}_p}{2} ,
\end{equation}
i.e., the maximally mixed state within the subspace free of (dressed) flux loops: an incoherent ensemble over all configurations of fermions on a flux-free background. On a closed manifold with nontrivial topology, this state additionally carries the projectors $(1+\lambda_i L_{\mathrm{f}}(\gamma_i))/2$ onto the eigenvalues of the non-contractible fermion loops $\gamma_i$ in $\rho_0$. The logical information carried by the fermion loops therefore survives as classical information, whereas the conjugate membrane logicals are completely randomized.

The symmetries of the maximally decohered state are:
\begin{enumerate}
    \item \textit{Strong $\mathbb{Z}_2$ 2-form symmetry.} The fermion loops $L_{\mathrm{f}}(\gamma)$ remain strong symmetries, with the 't~Hooft anomaly Eq.~\eqref{eq:resp1} encoding the fermionic statistics. Since open segments of $L_{\mathrm{f}}$ also commute with every Kraus operator, the self statistics of the line (and hence the anomaly) can be evaluated directly on the decohered state~\cite{Kobayashi2026statistics}.
    \item \textit{Weak $\mathbb{Z}_2$ 1-form symmetry.} The closed membrane
    \begin{equation}
        M(\Sigma) = \prod_{e \perp \Sigma} X_e = \prod_{v \in R} A_v , \qquad \Sigma = \partial R ,
    \end{equation}
    measures the $\mathbb{Z}_2$ gauge charge (i.e., the fermion parity) enclosed by $\Sigma$ and, in the parent state $\rho_0$, detects the fermion through braiding with the bosonic flux loop created by an open membrane. $M(\Sigma)$ commutes or anticommutes with each $\tilde{L}_e$, so $M(\Sigma)\,\mathcal{N}\rho_0\, M(\Sigma)^\dagger = \mathcal{N}\rho_0$ for all $p$. However,
    \begin{equation}
        \mathrm{Tr}\big[ M(\Sigma)\, \mathcal{N}\rho_0 \big] = (2p-1)^{|\Sigma|} ,
    \end{equation}
    where $|\Sigma|$ is the number of edges crossing $\Sigma$. This expectation value vanishes identically at $p=1/2$, so $M(\Sigma)$ is at most a weak symmetry at the maximally decohered point.
\end{enumerate}

The strong topological operators of $\mathcal{N} \rho_0$ are therefore generated by the fermion line $L_{\mathrm{f}}$ (together with the surfaces obtained by condensing it). These form $\tSVect$, in which the fermion is transparent: no strong operator remains that braids nontrivially with it. Decoherence has thus reduced the nondegenerate braided fusion 2-category $\mathcal{Z}(\tSVect)$ of the parent fermionic toric code to its maximally degenerate subcategory $\tSVect$. This is the direct (3+1)D analogue of the (2+1)D imTO with a single transparent fermion, $\{1,\psi\} \simeq \mathbf{SVect}$, obtained by decohering the fermion of the (2+1)D toric code~\cite{sohal2024imto,ellison2024,Wang2025intrinsic}.

Let us now argue that the maximally decohered state represents an \emph{intrinsically} mixed-state topological order. First, it is long-range entangled: since it is strongly symmetric under a $\mathbb{Z}_2$ 2-form symmetry with the anomaly Eq.~\ref{eq:resp1}, it cannot be prepared from a product state by any finite-depth local channel and has exponentially small overlap with every short-range entangled state~\cite{Hsin:2025pqb}. Note that the anomaly plays a central role: the bosonic counterpart of this construction, in which the \emph{bosonic} charges of the (3+1)D $\mathbb{Z}_2$ toric code are maximally decohered, produces $\rho \propto \prod_p (1+B_p)$, a classical ensemble of closed $\mathbb{Z}_2$ strings. That state can be prepared from $|0\rangle^{\otimes N}$ by applying each $A_v$ independently with probability $1/2$, and is therefore two-way connected to a trivial state (it is distinguished from a trivial state under a distinct equivalence relation that also accounts for anomalies between strong and weak symmetries~\cite{sang2024,sang2025lr,sohal2024imto,zhang2025}). Thus, the nontriviality of the maximally decohered state in our case hinges on the fermionic statistics of the proliferated particle.

The maximally decohered state cannot be two-way connected to the ground state of any gapped local Hamiltonian in three spatial dimensions. In such a ground state, remote detectability requires every nontrivial line operator to be detected by braiding with some surface operator. Here, however, the anomalous fermion line braids trivially with all strong operators. Categorically, the strong operators form the degenerate category $\mathbf{2SVect}$, whose transparent fermion carries the anomaly~\eqref{eq:resp1}. This is precisely the situation in which we expect two-way channel connectivity to a pure state to fail, and it can be made sharp using the refined notion of phase equivalence of Ref.~\cite{sang2025lr}: locally reversible channels preserve strong and weak higher-form symmetries together with their anomalies. In the fermionic toric code, the fermion line has a mutual anomaly with the \emph{strong} membrane $M(\Sigma)$, whereas here it braids nontrivially only with the \emph{weak} $M(\Sigma)$. Hence, the maximally decohered state realizes an imTO. It is the simplest member of the family with emergent fermions: in the language of Theorem~\ref{thm:mainB} it corresponds to $E$, $N$, and $H$ all trivial. It also serves as the building block for the FcFl theory of Sec.~\ref{subsec:FcFl}, whose layer $A$ is the state constructed here.

Finally, note that $L_{\mathrm{f}}(\gamma)$ is an exact strong symmetry for all $p$, so the anomaly-enforced long-range entanglement persists throughout $0 \le p \le 1$. What distinguishes weak from maximal decoherence is the fate of $M(\Sigma)$. For $p$ sufficiently close to $1$, we expect that the fermions can be recovered by a quasi-local decoder, so that $\mathcal{N}\rho_0$ remains two-way connected to the fermionic toric code (mTO). The imTO phase should then be separated from it by a decoding transition at some finite threshold.


\subsection{Review: Mixed topological order with a single fermionic loop}
\label{subsec:singlefloop}

We illustrate another example of a nontrivial imTO obtained by subjecting (3+1)D $\Z_2$ \textit{bosonic} toric code,  to a local noise channel. This is a review of an example that initially appeared in Ref.~\cite{Hsin:2025pqb}, which achieves a nontrivial imTO starting with a pure state topological order with only bosonic excitations. We stress that in contrast to pure state theories in (2+1)D, where one needs a fermionic quasiparticle to obtain an imTO by decoherence, this is not the case in (3+1)D.

We begin with the (3+1)D $\Z_2$ toric code. We prepare it with $\Z_4$ qudits living on each edge $e$ of the 3d cubic lattice:
\begin{align}
    H_{\text{TC}} = -\sum_e X_e^2 -\sum_v (A_v+A^\dagger_v) - \sum_p B^2_p ~,
    \label{eq:Z2TC_qudit}
\end{align}
which can be interpreted as the $\Z_4$ toric code  with the $m^2$ magnetic fluxes condensed, resulting in the $\Z_2$ toric code. Here, the $A_v$ and $B_p$ terms in the Hamiltonian are the standard vertex and plaquette terms of the (3+1)D $\Z_4$ toric code that involves the operators $X_e, Z_e$ satisfying the $\Z_4$ Pauli algebra $Z_e X_e = e^{2\pi i/4} X_e Z_e$.

This $\Z_2$ toric code hosts an anomalous emergent $\Z_2$ 1-form symmetry, generated by the closed membrane operator
\begin{align}
    S_{\mathrm{f}}(\Sigma) = \prod_{e\subset \Sigma} S_e~,
    \label{eq:floop}
\end{align}
where $\Sigma$ is a surface of the dual lattice, and the product is over edges $e$ crossing the surface $\Sigma$.
This 't Hooft anomaly corresponds to nontrivial statistics of the magnetic flux loop excitation dubbed a fermionic loop.

We then consider an error channel which preserves the above anomalous $\Z_2$ symmetry, and leads to the incoherent proliferation of such fermionic loop excitations:
\begin{align}
    \mathcal{N} = \prod_e\mathcal{N}_e, \quad \mathcal{N}_e(\rho) = p\rho + (1-p) \tilde S_e \rho \tilde S_e^\dagger~,
\end{align}
where $\tilde S_e$ is defined in Figure~\ref{fig:Se} (b). In the case of maximal decoherence $p=1/2$, the Choi state for the decohered state is expressed as
\begin{align}
    \ket{\mathcal{N}\rho_0}\rangle = \prod_e\left( \frac{1 + \tilde S^+_e (\tilde S^{-}_{e})^*}{2} \right) \ket{\text{TC}}_+\ket{\text{TC}}_-~,
\end{align}
where we took the initial state $\rho_0$ to be a pure ground-state of the (3+1)D toric code, $\rho_0 = \ket{\text{TC}}\bra{\text{TC}}$.

The maximally decohered mixed state is invariant under the strong (emergent) anomalous $\Z_2$ 1-form symmetry $S_{\text{f}}(\Sigma)$, in the following sense:
\begin{enumerate}
    \item The closed (oriented) surface operators $S_\text{f}(\Sigma)$ commute both with every $\tilde S_{e}$ and with the toric code Hamiltonian. Therefore, it remains a strong symmetry of the decohered state:  $S_{\text{f}}(\Sigma) \mathcal{N}\rho_0 =\mathcal{N}\rho_0$.
    \item When the surface $\Sigma$ is not closed, $S_{\text{f}}(\Sigma)$ does not preserve $\mathcal{N}\rho_0$ at the boundary. Meanwhile, its square commutes with every $\tilde S_e$ and with $H_\text{TC}$ even when $\Sigma$ has a boundary. Therefore, $(S_{\text{f}}(\Sigma))^2 \mathcal{N}\rho_0 =\mathcal{N}\rho_0$ i.e., this remains a strong symmetry even when $\Sigma$ is not closed. This implies that the defect of the 1-form symmetry follows the $\Z_2$ fusion rule on the state $\mathcal{N}\rho_0$, and generates the emergent $\Z_2$ symmetry. This $\Z_2$ fusion rule of the symmetry defect allows us to define the generalized statistics invariant~\cite{Kobayashi2026statistics} and define the 't Hooft anomaly associated with the emergent $\Z_2$ 1-form symmetry. The invariant carried by $S_{\text{f}}$ turns out to correspond to a fermionic loop statistic~\cite{Hsin:2025pqb}.
    \end{enumerate}

The above strong $\Z_2$ 1-form symmetry has fermionic loop statistics, which corresponds to an 't Hooft anomaly characterized by the (4+1)D response
\begin{equation}
\pi\int B\cup Sq^1 B=    \pi\int w_3\cup B~,
\label{eq:BSq1B}
\end{equation}
using the 2-form $\Z_2$ background gauge field $B$. This 't Hooft anomaly of the strong symmetry enforces the long-range entanglement of the mixed state $\mathcal{N}\rho_0$, and protects a nontrivial mixed-state phase.  The Wilson loops $\prod_{e\in\gamma}Z_e^2$ of the parent toric code, which detected the flux loop by braiding, anticommute with $\tilde S_e$ on edges $e \in \gamma$ and hence survive only as weak symmetries. The strong operators therefore form $\mathbf{2Vect}^{\beta}_{\mathbb{Z}_2}$ with $\beta$ the fermionic-loop class: the fermionic loop is transparent and anomalous, placing this example in the first class of imTOs (cf.\ Sec.~\ref{subsection:onlyboson}, where the $\beta=0$ counterpart is shown to be classical).

\subsection{Mixed topological order with a non-Abelian fermionic loop}
\label{subsec:nonabelianfloop}

Here, we describe yet another (3+1)D imTO with a \textit{non-Abelian} fermionic loop obtained by sequentially applying local noise channels. This theory was previously discussed by one of us in Ref.~\cite{sato2026invariants}. A non-Abelian fermionic loop is a non-invertible defect that exhibits fermionic loop statistics. We note that the non-Abelian fermionic loop excitation is found in $\mathbb{D}_4$ gauge theory (dihedral group of order 8) in (3+1)D~\cite{sato2026invariants}. On a cubic lattice, we start with two copies of bosonic $\Z_2$ toric code ground states $\rho_0^A, \rho_0^B$, together with a product state $\ket{0}\bra{0}^C$ with additional qubits on edges. We then add a local noise channel $\mathcal{N}^A, \mathcal{N}^B$ on each layer of $\Z_2$ toric codes, so that each layer yields an imTO with a single fermionic loop. We then get a mixed state in the form of
\begin{equation}
\rho_1=\mathcal{N}^{A}\rho_0^{A}\otimes \mathcal{N}^{B}\rho_0^{B}\otimes \ket{0}\!\bra{0}^{C},
\end{equation}
A straightforward approach to obtaining a non-Abelian theory is to gauge the $\mathbb{Z}_2$ layer swap symmetry between a pair of Abelian gauge theories (see e.g. Refs~\cite{bulmashgauging,premgauging}). Motivated by this, we thus gauge the weak $\mathbb{Z}_2$ swap symmetry exchanging $A,B$ by a vertex-local channel
\begin{equation}
\mathcal{N}'=\prod_v \mathcal{N}'_v,\qquad
\mathcal{N}'_v(\rho)=p'\,\rho+(1-p')\,G_v\rho G_v^\dagger,
\end{equation}
with local Gauss operator of $\Z^C_2$ SWAP symmetry
\begin{equation}
G_v=A_v^{(C)}\,\mathrm{SWAP}^{(AB)}_{N(v)}\mathrm{SWAP}^{(AB)}_{E(v)}\mathrm{SWAP}^{(AB)}_{U(v)}~,
\end{equation}
where $A_v^{(C)}=\prod_{v\subset\partial e} X^C_e$ is the $\mathbb{Z}^C_2$ gauge-field star term in the $C$ sector.  Here, $N(v)$, $E(v)$, and $U(v)$ denote the edges emanating from $v$ in the $+\hat{x}$, $+\hat{y}$, and $+\hat{z}$ directions, respectively (so that each edge is assigned to a unique vertex), and $\mathrm{SWAP}^{(AB)}_e$ exchanges the $A$ and $B$ qudits on edge $e$. At the maximally decohered point $p'=p=\tfrac12$, the mixed state $\rho=\mathcal{N}'\rho_1$ has a gauge-invariant non-invertible strong symmetry operator with fermionic loop statistics~\cite{sato2026invariants}. This symmetry operator is understood as a gauge invariant sum of fermionic loop $\Z_2$ 1-form symmetry operators from $A, B$ layers, which corresponds to coset 1-form symmetry under gauging the $\Z_2^C$ SWAP symmetry~\cite{Hsin2024cosine, Hsin2026coset}.

The above mixed state defines a (3+1)D imTO with intrinsic long-range entanglement, enforced by the non-Abelian fermionic loop statistics. In the language of Sec.~\ref{subsection:egs}, this theory corresponds to $E=\mathbb{Z}_2\times\mathbb{Z}_2$ with fermionic-loop anomalies in both factors, $N$ trivial, and $H=\mathbb{Z}_2$ acting weakly by the swap. Weak gauging produces a transparent, anomaly-free $\mathbb{Z}_2$ Wilson line of the $C$ gauge field, while the intrinsically mixed character stems from the anomalous, transparent non-Abelian fermionic loop.

\begin{figure}[htb]
\centering
\includegraphics[width=0.45\textwidth]{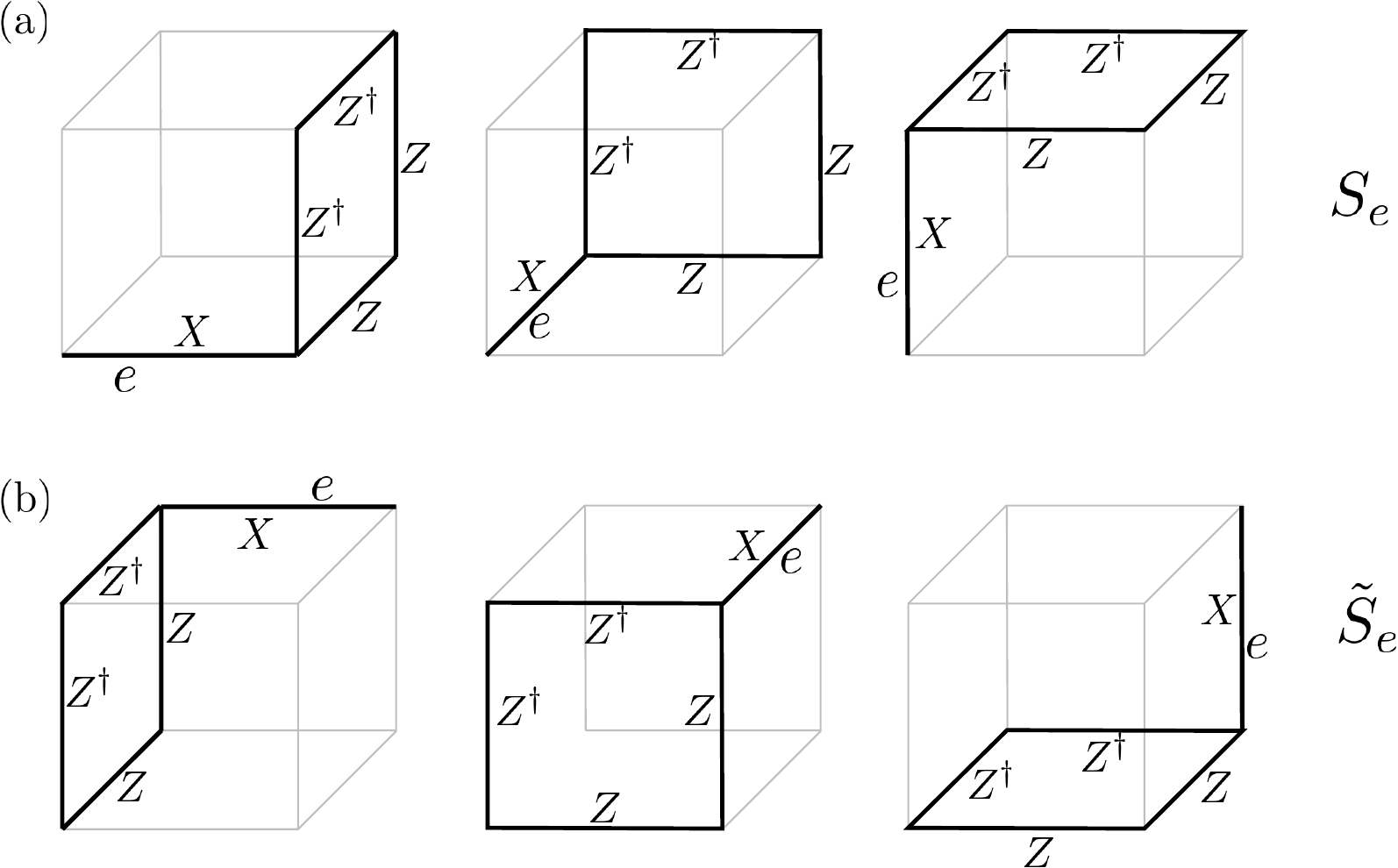}
\caption{Definitions of the operators $S_e, \tilde S_e$.}
\label{fig:Se}
\end{figure}
\subsection{Mixed topological order with an emergent fermion and fermionic loop: FcFl theory}
\label{subsec:FcFl}
Finally, we provide a new example of (3+1)D imTO that consists of a single fermionic particle and single fermionic loop with nontrivial mutual braiding. These topological operators generate the strong symmetry of the mixed state; the strong symmetry corresponds to the (3+1)D $\Z_2$ gauge theory with global gravitational anomaly (dubbed an FcFl theory~\cite{FHH21}), characterized by the fermionic statistics carried by both the particle and loop excitation.

To obtain the desired mixed state, we start with a stacked system of a bosonic $\Z_2$ and a fermionic $\Z_2$ toric code,
\begin{align}
    H = H^A_{\text{fTC}} \otimes \mathds{1}^B + \mathds{1}^A \otimes H^B_{\text{TC}}~,
\end{align}
where the Hamiltonian for each layer $A,B$ is given by \eqref{eq:fTC}, \eqref{eq:Z2TC_qudit}. 
Its ground state is denoted by $\rho_0 =\ket{\text{fTC}}\bra{\text{fTC}}  \otimes \ket{\text{TC}}\bra{\text{TC}}$. We then apply the quantum channel
\begin{align}
    \mathcal{N}= \prod_e \mathcal{N}^A_e  \mathcal{N}^{AB}_e~,
\end{align}
with
\begin{align}
\begin{split}
    \mathcal{N}^A_e(\rho) &= p\rho + (1-p) \tilde L^A_e \rho  \tilde L_e^{A\dagger}~, \\
    \mathcal{N}^{AB}_e(\rho) &= p\rho + (1-p) X^A_e \tilde S^B_e \rho  X^A_e\tilde S_e^{B\dagger}~,
    \end{split}
\end{align}
where the operators $\tilde L_e, \tilde S_e$ on each layer are shown in Figs.~\ref{fig:Le}, \ref{fig:Se}.
The state $\mathcal{N}\rho_0$ then has an anomalous $\Z_2$ 2-form symmetry generated by a line operator
\begin{align}
    L^{AB}_{\text{f}}(\gamma) = \prod_{e\subset\gamma} L^A_{e}(Z_e^B)^2~,
\end{align}
as well as an anomalous $\Z_2$ 1-form (emergent) symmetry generated by a surface operator
\begin{align}
    S^B_{\mathrm{f}}(\Sigma) = \prod_{e\subset \Sigma} S^B_e~.
\end{align}
These operators have nontrivial mutual braiding. In contrast, the fermion line of layer $A$ and the Wilson line of layer $B$ are individually only weak symmetries (each anticommutes with $X^A_e\tilde S^B_e$) as is the membrane of layer $A$, which anticommutes with $\tilde L^A_e$. Only their combination $L^{AB}_{\rm f}$ survives as strong.

The above strong symmetry generates the FcFl theory since it consists of a fermionic particle and a fermionic loop with mutual braiding. In contrast to the examples of Secs.~\ref{subsec:singlefermion}--\ref{subsec:nonabelianfloop}, the strong operators here form a \emph{non-degenerate} braided fusion 2-category, namely $\mathcal{T}$ (see Sec.~\ref{subsection:structureclassification}): the fermion and the fermionic loop detect each other by braiding, so remote detectability is not violated. The state is nevertheless intrinsically mixed, since $\mathcal{T}$ carries the gravitational anomaly $\pi\int w_2 w_3$ and is not expected to be realizable in the ground state of any gapped local Hamiltonian in strictly three spatial dimensions~\cite{FHH21,Chen2023gravitational}. Its long-range entanglement is enforced by the anomalous strong symmetries~\cite{Hsin:2025pqb}. In the next Section, we explain how this theory fits into a general class of theories which serve as ``seeds'' for constructing more intricate imTOs.

\section{General construction of (3+1)D mixed topological order}\label{section:generalconstruction}
The above lattice model examples are constructed by the following procedure:
\begin{enumerate}
    \item We first start with a pure state for copies of $\Z_N$ toric codes (with bosonic/fermionic quasiparticles), and turn on an error channel to obtain a mixed state $\rho$ with strong $\Z_N$ 1-form/2-form symmetries. The strong symmetry of the resulting mixed state forms a braided fusion 2-category which is either degenerate or non-degenerate but anomalous. For instance, the strong symmetry of the model in Sec.~\ref{subsec:singlefermion} consists of a single fermionic quasiparticle, and that in Sec.~\ref{subsec:singlefloop} consists of a single fermionic loop, both of which are transparent and form a degenerate theory. Meanwhile, the strong symmetry of the model in Sec.~\ref{subsec:FcFl} forms the non-degenerate but anomalous theory $\mathcal{T}$.
    
    \item The above mixed state $\rho$ then has a  0-form symmetry, which may be strong or weak. We obtain a new imTO by gauging these strong/weak 0-form symmetries. For instance, the model in Sec.~\ref{subsec:nonabelianfloop} is obtained by gauging the weak $\Z_2$ SWAP symmetry of the two copies of intrinsically mixed $\Z_2$ toric codes, producing a theory with a non-Abelian fermionic loop.
\end{enumerate}

The above construction motivates a general procedure to achieve imTOs in (3+1)D. In fact, a suitable generalization, considering all algebraic structures of finite strong symmetries in (3+1)D TOs, yields a general construction which, as we show in Sec.~\ref{section:allbraided}, produces all braided fusion 2-categories.

\subsection{General construction: theories with only bosonic quasiparticles}\label{subsection:onlyboson}
We first describe the general construction of mixed state topological orders that contain only bosonic quasiparticles:
\begin{enumerate}
    \item We start with a bosonic toric code with an Abelian gauge group $E$ in (3+1)D. We turn on a quantum channel that incoherently proliferates the bosonic/fermionic loop excitation, which corresponds to the $E$ 1-form symmetry with an 't Hooft anomaly $\beta\in\rH^5(B^2E; \mathbb{C}^\times)$. 
    The 't Hooft anomaly describes the bosonic/fermionic loop statistics of the proliferated excitation.
    The resulting mixed state $\rho_{(E,\beta)}$ carries mTO, described by a degenerately braided strongly fusion 2-category $\fB =\tVect^\beta_E$ (see Theorem \ref{thm:mainA}): the strong $E$ 1-form symmetry is now an input for the grading of the category of vector spaces which labels the surface operators. The 't Hooft anomaly $\beta$ controls the braiding of surface operators given two ambient dimensions.
    In the categorical language, the example in Sec.~\ref{subsec:singlefloop} starts with a pure TO described categorically by $\cZ(\tVect_{\Z_2})$, and the imTO 
    corresponds to $E=\Z_2$, with $\beta$ the nontrivial element of $\rH^5(B^2\Z_2;\mathbb{C}^\times)=\Z_2$.
    
    For $\beta$ nontrivial, $\rho_{(E,\beta)}$ is an imTO. For $\beta=0$, in contrast, the proliferated loops are anomaly-free bosons, and the maximally decohered state (e.g., $\rho\propto\prod_v(1+A_v)$ for $E=\mathbb{Z}_2$) is a classical membrane ensemble which can be prepared from a product state by a finite-depth channel, and is hence two-way connected to a trivial state.

    \item Suppose the resulting mixed state $\rho_{(E,\beta)}$ has a weak 0-form finite symmetry $G$, and its normal subgroup $N\subset G$ generates a strong symmetry.
    We allow $G$ to act as an anomalous 0-form symmetry on the braided strongly fusion 2-category $\fB$ formed by the strong 1-form symmetry of $\rho_{(E,\beta)}$, while we require $N$ to act by non-anomalous symmetry.  We gauge the strong $N$ symmetry, and then gauge the weak $G/N$ symmetry. At the level of the Choi state $\ket{\rho_{(E,\beta)}}\rangle$, this corresponds to gauging the $[(N \times N)\rtimes G_{\text{diag}}]/N_{\text{diag}}$ symmetry.    
\end{enumerate}

\subsection{General construction: theories that contain fermionic quasiparticles}
\label{subsection:emergentfermion}

We then provide the  general construction of mixed state topological orders that contain fermionic quasiparticles.

\begin{enumerate}
    \item We consider a mixed state $\rho_{(E,\varsigma,\kappa)}$ that generalizes $\rho_{(E,\beta)}$ in the bosonic case to the case with an emergent fermion. The state $\rho_{(E,\varsigma,\kappa)}$ has the following strong symmetry:
    \begin{itemize}
        \item $\rho_{(E,\varsigma,\kappa)}$ has a strong $\Z_2$ 2-form symmetry that corresponds to the line operator of an emergent fermion. The fermionic statistics corresponds to the 't Hooft anomaly with the 5D response $\pi \int Sq^2 A_3$ using the 3-form background $A_3\in Z^3(M, \Z_2)$.
        \item $\rho_{(E,\varsigma,\kappa)}$ has a strong $E$ 1-form symmetry with $E$ a finite Abelian group. The 1-form symmetry operator has a mutual braiding with the emergent fermion, which is characterized by the degree 2 class $\kappa\in\rH^2(B^2 E;\Z_2)$; it corresponds to the mixed 't Hooft anomaly between 1,2-form symmetries with the (4+1)D response $\pi \int (B^*\kappa) \cup A_3$ using the background gauge fields $B: M\to B^2 E$, $A_3\in Z^3(M, \Z_2)$. $\varsigma$ then labels the self anomaly of the 1-form symmetry $E$; it is an element of the supercohomology of degree 5, $\varsigma \in \SH^{5+\kappa}(B^2E)$. This supercohomology classifies the 't Hooft anomaly for $E$ 1-form symmetry of the theory equipped with the twisted spin structure, whose obstruction class is $w_2 + B^*\kappa$ (see~\cite[Definition 2.10]{Debray:2023iwf} for the full definition). The class $\varsigma$ is understood as the 't Hooft anomaly of the $E$ 1-form symmetry after gauging the $\Z_2$ 2-form symmetry of the state $\rho_{(E,\varsigma,\kappa)}$; the resulting theory after gauging the 2-form symmetry has the 1-form symmetry $E$ and depends on the twisted spin structure~\cite{kobayashi2026sixteenfoldway}, where $E$ has the 't Hooft anomaly $\varsigma$.
    \end{itemize}

    The strong symmetry of the state $\rho_{(E,\varsigma,\kappa)}$ forms a braided fermionic fusion 2-category $\fB =\tSVect^{(\kappa,\varsigma)}_{E}$ (see Theorem \ref{thm:mainB}).
    For instance, the example of the FcFl  imTO constructed in Sec.~\ref{subsec:FcFl} corresponds to $\rho_{(E,\varsigma,\kappa)}$ with $E=\Z_2$, $\kappa\in\rH^2(B^2\Z_2;\Z_2)=\Z_2$ being the nontrivial element (so that the 1, 2-form symmetry have nontrivial mutual braiding), and $\varsigma\in\SH^{5+\kappa}(B^2\Z_2)=\Z_2$ being the nontrivial element that encodes fermionic loop statistics of $\Z_2$ 1-form symmetry.
    
     \item The 0-form gauging then proceeds analogously to step 2 of Sec.~\ref{subsection:onlyboson}: the non-anomalous normal subgroup $N\trianglelefteq G$ is gauged as a strong symmetry, followed by the weak quotient $H=G/N$. Specifically, suppose the mixed state $\rho_{(E,\varsigma,\kappa)}$ has a weak 0-form finite symmetry $G$, and its normal subgroup $N \trianglelefteq G$ generates a strong symmetry. We allow $G$ to act as an anomalous 0-form symmetry on the braided strongly fusion 2-category $\fB$ formed by the strong 1-form symmetry of $\rho_{(E,\varsigma,\kappa)}$, while we require $N$ to act by non-anomalous symmetry.  We gauge the strong $N$ symmetry, and then gauge the weak $G/N$ symmetry.
    At the level of the Choi state $\ket{\rho_{(E,\varsigma,\kappa)}}\rangle$, this corresponds to gauging the $[(N \times N)\rtimes G_{\text{diag}}]/N_{\text{diag}}$ symmetry.  
\end{enumerate}

\section{Mixed topological orders realize all braided fusion 2-categories}
\label{section:allbraided}
The theories presented thus far reveal a close connection between the topological operators of imTOs and braided fusion 2-categories. By adopting the following ansatz for mTOs, we place this relationship on a systematic footing, enabling more general statements about their characterization using category theory. This will be the focus of the remainder of the paper.

\begin{ansatz}
\label{ansatz:mTO}
The strong topological operators of a (3+1)D mixed state topological order form a braided fusion 2-category.
\end{ansatz}

\begin{notation*}
 In this Ansatz, and for all the examples in this paper, we are only concerned with \textit{bosonic} topological orders, which amounts to the fact that the theory does not need to couple to a spin structure to be defined. Even though the topological order is bosonic, there can still be emergent fermions in the spectrum of operators as discussed in Sec. \ref{sec:latmod}. To distinguish whether or not a theory has an emergent fermion we use the following terminology. In the case when all line operators of a bosonic topological order are bosons (loops may still carry fermionic statistics), the theory will be described by a ``bosonic'' braided fusion 2-category. In the case when the spectrum contains emergent fermions, the theory will be described by a ``emergently-fermionic'' (EF) braided fusion 2-category. We will not be concerned with inherently fermionic (i)mTOs in this paper, which require a spin structure to be defined. 
\end{notation*}

The  definition of a  braided fusion 2-category with emergent fermions, along with details on $G$-crossed extensions are given in Appendix~\ref{appendix:F2C}. In particular, the reader is directed to Sec.~\ref{subsection:enrichment} therein for notation and definitions.   With this connection presented in the ansatz in place, we present the parametrization of mTOs by way of the data parameterizing braided fusion 2-categories.

Recall that a \textit{supergroup} $(H, z)$ is a group $H$ with a central element $z \in H$ of order two. We let $\Rep(H, z)$ denote the symmetric
monoidal category with underlying monoidal category $\Rep(H)$ and braiding twisted by the action
of $z$. A result of Deligne \cite{deligne2002} shows that  $\Rep(H, z)$ are all the symmetric monoidal categories which are not Tannakian. These will play an important role in the study of EF braided fusion 2-categories.

 \begin{defn}[\cite{JFY}]\label{def:strongly}
     A bosonic (resp. EF) braided fusion 2-category $\fC$ with monoidal unit $\mathds{1}$ is a braided fusion 2-category such that $ \Omega \fC:=\End_{\mathds{1}}(\fC)$ is $\Rep(G)$ (resp. $\Rep(G,z)$), for some finite group $G$. 
 \end{defn}

\begin{theorem}
\label{thm:mainA}
Mixed state topological orders in (3+1)D whose point-like excitations are all bosons correspond to braided fusion 2-categories parametrized, up to equivalence, by the following data:
    \begin{enumerate}
        \item \label{item:A1}A finite abelian group $E$;
    \item \label{item:A2} A class $\beta \in \rH^5(B^2 E; \mathbb{C}^\times)$;
    \item \label{item:A3} A finite group N;
    \item \label{item:AmaptoPic}  A faithfully graded $N$-crossed braided extension $\fC$ of $\fB = \tVect^\beta_E$;
    \item \label{item:A5} A finite group $H$;
    \item \label{Gaction} An action $\rho:B H \rightarrow B\sAut^{\br}(\fC^N)$, where $\sAut^{\br}(\fC^N)$ denotes the 3-group of braided autoequivalences of the equivariantization of $\fC$ by $N$.
    \end{enumerate}
\end{theorem}

Let us summarize how this parametrization reflects the construction in Sec. \ref{subsection:onlyboson}. The braided categories of the form $\mathfrak{B}=\mathbf{2Vect}^\beta_E$ are the mTOs (intrinsically mixed for $\beta\neq0$) constructed from the decohered toric codes with gauge group $E$ and 't Hooft anomaly $\beta$. We would like to gauge a weak $G$-symmetry, which is a symmetry of the mTO to produce one with nontrivial line operators i.e. 2-form symmetries. As stressed in the second point in Sec. \ref{subsection:onlyboson}, we allow $N\subset G$ to act non-anomalously as a strong symmetry, and gauge the $N$ symmetry. This is precisely captured by the two step procedure: take a faithfully graded  $N$-crossed braided extension of $\fB$ given by $\fC$ and equivariantize by $N$. We will elaborate more on this in the next section.  Finally, we have $H = G/N$ as the left over weak symmetry that acts on $\fC^N$ which we can gauge. The combined operation was described on the Choi state in Sec. \ref{subsection:onlyboson}. 

In a similar vein as Theorem \ref{thm:mainA}, we give the data which realizes all mTOs (intrinsic or not) that contain an emergent fermion.
\begin{theorem}\label{thm:mainB}
    Mixed state topological orders in (3+1)D whose point-like excitations include an emergent fermion correspond to braided fusion 2-categories parametrized, up to equivalence, by the following data:
    \begin{enumerate}
        \item \label{item:B1} A finite abelian group $E$;
        \item \label{item:B2} A class $\kappa \in \rH^2(B^2 E; \mathbb{Z}_2)$ and a  supercohomology  class $\varsigma \in \SH^{5+\kappa}(B^2 E)$\footnote{We refer the reader to Appendix \ref{subsection:fermionicbackground} for the definition of supercohomology.};
        \item A finite group $N$;
        \item  \label{item:BmaptoPic} A faithfully graded $N$-crossed braided extension $\fC$ of $\fB = \tSVect^{(\kappa,\varsigma)}_{E}$, a braided fusion 2-category with $\SVect$ as the category of one-morphisms;
        \item \label{item:groupH} A finite group $H$;
        \item An action $\rho:B H\rightarrow B\sAut^{\br}(\fC^N)$.
    \end{enumerate}
\end{theorem}

\textit{Pure-state versus intrinsically mixed.} By Theorem~\ref{thm:equiv}, equivariantizing a $G$-crossed braided fusion 2-category yields a non-degenerate braided fusion 2-category if and only if the grading is faithful and the trivially graded component is non-degenerate. In the all-bosons case, $\tVect^\beta_E$ has no nontrivial line operators, so all its surface operators are transparent and it is non-degenerate only for $E$ trivial. Pure-state (3+1)D TOs with only bosonic point-like excitations therefore arise precisely from faithfully graded extensions of $\mathbf{2Vect}$, recovering the classification of Ref.~\cite{PhysRevX.8.021074}. In the emergent-fermion case, $\mathbf{2SVect}^{(\kappa,\varsigma)}_E$ is non-degenerate only for $E=\mathbb{Z}_2$ with $\kappa$ nontrivial, i.e.\ for $\mathcal{S}$ or $\mathcal{T}$; since $\mathcal{T}$ carries the gravitational anomaly $\pi\int w_2w_3$, pure states arise from faithfully graded extensions of $\mathcal{S}$. All other data describe mixed states. We expect these to be intrinsically mixed precisely when the strong operators carry an anomaly that cannot be realized in a (3+1)D ground state, which occurs in two ways: (i) the transparent operators are anomalous, as for the fermion of Sec.~\ref{subsec:singlefermion} and the fermionic loops of Secs.~\ref{subsec:singlefloop} and~\ref{subsec:nonabelianfloop}; or (ii) the trivially graded component is $\mathcal{T}$, as for the FcFl theory of Sec.~\ref{subsec:FcFl}. When the degeneracy instead stems only from anomaly-free bosonic transparent operators (e.g., $E$ nontrivial with $\beta=0$), the corresponding states are essentially classical ensembles which, under two-way channel equivalence, are connected to a pure state. Establishing this criterion in general is left to future work.

\subsection{Structure of Mixed Topological Order}\label{subsection:structureclassification}
We now explain in detail from where the data of braided fusion 2-categories arises, and also their categorical construction. If $\fD$ is a bosonic braided fusion 2-category, then $\Omega \fD$ will be a Tannakian symmetric fusion 1-category, say $\mathbf{Rep}(G)$, for some finite group $G$. On the other hand, if $\fD$ is a EF braided fusion 2-category, then $\Omega \fD$ is a super-Tannakian symmetric fusion 1-category of the form $\Rep({G},z)$, for some super-group $({G},z)$. In particular, it contains $\Rep({G})$ with $G=({G},z)/z$ as a symmetric tensor fusion sub-1-category. In either case, we find that there is a canonical braided tensor 2-functor $\mathbf{2Rep}({G})\rightarrow\mathfrak{D}$, which we can use to de-equivariantize $\mathfrak{D}$. The end result is a  $G$-crossed braided strongly fusion 2-category,  which is either bosonic or emergently fermionic depending on the input $\mathfrak{D}$. An important detail to notice here is that this category need not be faithfully $G$-graded. In essence, the difference between mTOs and pure TOs in (3+1)D lies in the fact that pure-state TOs always arise from faithfully $G$-graded $G$-crossed braided fusion 2-categories.
From a $G$-crossed braided strongly fusion 2-category, we can equivariantize by the $G$-action to recover $\mathfrak{D}$, hence inverting the process, and this can be interpreted as reinstating the line operators. It therefore suffices to classify non-faithfully graded $G$-crossed braided strongly fusion 2-category to classify braided fusion 2-categories and thereby partially classify (3+1)D mTOs.

\subsubsection{The All Bosons Case}
\label{subsection:allboson}
We begin by giving the data of a braided bosonic strongly fusion 2-category, and then working up to the $G$-crossed braided extension. By applying Proposition~\ref{prop:braidedextn} to the ``vacuum theory'' i.e. $\tVect$, we find that braided bosonic strongly fusion 2-categories are classified by a finite abelian group $E$
together with a class $\beta \in\rH^5(B^2 E; \mathbb{C}^\times)$. This can be seen explicitly by noting that the space that classifies faithfully $E$-graded braided extensions of $\tVect$, by Proposition~\ref{prop:braidedextn}, is $B^2 \sPic^{\br}(\tVect)$ which has only a single nontrivial homotopy group, $\pi_5(B^2\sPic^{\br}(\tVect)) = \mathbb{C}^\times$.

Physically, the surface operators are one-form symmetries with fusion rules given by the abelian group $E$. The class in $\rH^5(B^2 E; \mathbb{C}^\times)$ parametrizes grouplike 2-monoidal 2-categories, hence gives a braiding to the operators. Let $\fB=\tVect^\beta_E$ denote the braided bosonic strongly fusion 2-categories corresponding to the aforementioned data. The data of the map in Item \ref{item:AmaptoPic}  in Theorem \ref{thm:mainA} requires the fact that the homotopy groups of $\sPic^{}(\tVect^\beta_{E})$ are given by:
\begin{equation}\label{eq:pic2Vec}
\begin{tabular}{|c|c|c|c|c|}
\hline
$\pi_0$ & $\pi_1$ & $\pi_2$ & $\pi_3$\\
\hline 
$*$ & $E$ &  $0$ &  $\mathbb{C}^{\times}$\\
\hline
\end{tabular}
\end{equation}
where the single nontrivial $k$-invariant is given by $\beta$.

The groups $\pi_{\geq 1}(\sPic^{}(\tVect^\beta_{E}))$ as well as the $k$-invariant can be read off directly from the fact that, by construction, 
there is a fiber sequence 
\begin{equation}
    B \sPic(\tVect^\beta_{E}) \longrightarrow B^2 E \xlongrightarrow{\beta} B^5\mathbb C^\times\,.
\end{equation}
We now only have to argue that $\pi_{0}(\sPic^{}(\tVect^\beta_{ E}))$ is trivial. But it was shown in \cite[Proposition 3.4.1]{Decoppet:2023bay} that every invertible $\tVect^\beta_{ E}$-bimodule 2-category is given by an autoequivalence of $\tVect^\beta_{ E}$. It follows immediately that every invertible $\tVect^\beta_{ E}$-module 2-category must be trivial as asserted.

In building a non-faithfully $G$-graded $G$-crossed braided fusion 2-category out of $\fB$ we can first consider the part which is faithfully $N$-graded, where $N$ is a subgroup of $G$. 

\begin{lemma}\label{lem:nonfaithful-crossed}
Let $\fB$ be a braided fusion $2$-category, and let $N \trianglelefteq G$ be a normal subgroup. Faithfully $N$-graded $G$-crossed braided extensions of $\mathcal B$ are classified by the following data:
\begin{enumerate}
    \item A faithfully $N$-graded $N$-crossed braided extension
    $\fC$ of $\fB$;

    \item An action of $G$ on $\fC$ preserving the crossed
    braiding and inducing conjugation on the grading;

    \item An identification between the two induced actions of
    $N$ on $\fC$.
\end{enumerate}
\end{lemma}
\begin{proof}
Let $\mathfrak{C}$ be a faithfully $N$-graded $G$-crossed braided extension of $\mathfrak{B}$. In particular, we have that $\mathfrak{C}$ is equipped with an action by $G$ that preserves the crossed braiding and acts by conjugation on the grading. Restricting this action to $N$, we find that $\mathfrak{C}$ is a faithfully graded $N$-crossed braided fusion 2-category. These two structures have to be compatible in the obvious way. Conversely, given the data (1)--(3), the $G$-action on $\mathcal{C}$ together with the identification (3) promotes the $N$-crossed braiding to a $G$-crossed braiding, yielding a faithfully $N$-graded $G$-crossed braided extension of $\mathfrak{B}$. 
\end{proof}

If $\fD$ is a braided fusion 2-category with $\Omega \fD = \Rep(G)$, one can consider $\Omega \cZ_{(2)}(\fD)$ which is a symmetric fusion subcategory of $\Omega\fD$. See  Definition \ref{def:syllepticcenter} for the details on the sylleptic center $\cZ_{(2)}(-)$. In this case $\Omega \cZ_{(2)}(\fD) = \Rep(H)$, with $G\twoheadrightarrow H$ and we can take $N$ to be the kernel of the projection so that there is a short exact sequence 
\begin{equation}
    1\rightarrow N \rightarrow G \rightarrow H \rightarrow 1\,.
\end{equation}

  The categories $\fC$ in the first item of Lemma~\ref{lem:nonfaithful-crossed} are classified by homotopy classes of maps $B N \to B\sPic(\fB)$. Since $N$ serves as a faithful grading, this implies that we are extending  $\fB$ by a symmetry $N$ that is non-anomalous.
 Such a faithfully graded $N$-crossed braided extension of $\fB$ is given by $\fC = \oplus_{g\in N} \fB_g$, where $\fB_e = \fB$ and  $\fB_g$ is nontrivial for each element of $g\in N$. In particular, the resulting category also has an action by $N$ so it makes sense to take its equivariantization by $N$. This gives the category which is denoted $\fC^N$, and $\Omega \fC^N = \Rep(N)$. Finally, Item \ref{Gaction} of Theorem \ref{thm:mainA} endows $\fC^N$ with a $H$-action by braided autoequivalences. 
 Through equivariantizing the $H$-action on $\fC^N$, i.e. classically gauging, we can obtain a braided fusion 2-category $\fD$ with non-trivial line operators labeled by $\Rep(G)$.
In order to ensure that $\Omega\fD$ is indeed $\Rep(G)$ the action of $H$ on $\mathfrak{C}^N$ must act on $\Omega\mathfrak{C}^N=\mathbf{Rep}(N)$ so that $\mathbf{Rep}(N)^H = \mathbf{Rep}(G)$.

This concludes the justification of why each piece of data of Theorem \ref{thm:mainA} is necessary to classify braided fusion 2-categories, and hence characterizes (and partially classifies) bosonic (3+1)D mTOs:
 \begin{itemize}
     \item The data of Items~\ref{item:A1} and~\ref{item:A2} come from the data that goes into specifying $\fB$.
     \item Item~\ref{item:A3} specifies the normal subgroup $N\trianglelefteq G$ of the 0-form symmetry that acts strongly.
     \item The data of Item~\ref{item:AmaptoPic} encodes the strong $N$-symmetry, which must act non-anomalously in order to be gauged.
     \item The data of Item~\ref{item:A5} gives the group $H$, which records the non-faithfulness of the $G$-grading. In this sense $H$ can act anomalously i.e. it is the weak symmetry used to define the mTO.
     \item Item ~\ref{Gaction} specifies how the weak symmetry $H$ acts on the strong sector $\mathcal{C}^N$, including how $\mathbf{Rep}(N)$ is extended to $\mathbf{Rep}(G)$.
 \end{itemize}
 Consequently, a $G$-crossed braided fusion 2-category with trivially graded component $\mathbf{2Vect}^\beta_E$ equivariantizes to a non-degenerate braided fusion 2-category (and hence can describe a pure-state TO) only when the grading is faithful \emph{and} $E$ is trivial, consistent with the criterion stated after Theorem~\ref{thm:mainB}.
 
\subsubsection{The Emergent Fermionic Case}
We now consider the case when the spectrum has an emergent fermion. Following the logic outlined in the beginning of Sec. \ref{subsection:structureclassification}, we now employ Proposition~\ref{prop:braidedextn} to classify braided fermionic strongly fusion 2-categories. Such categories arise as $E$-graded braided extensions of $\tSVect$:
 \begin{lemma}
 \label{eq:piPicbr}
The homotopy groups of $\sPic^{\br}(\tSVect)$ are given by
\begin{equation}\label{eq:picbr2svect}
\begin{tabular}{|c|c|c|c|c|}
\hline
$\pi_0$ & $\pi_1$ & $\pi_2$ & $\pi_3$\\
\hline 
$\Z_2$ & $\Z_2$ &  $\Z_2$ &  $\mathbb{C}^{\times}$\,. \\
\hline
\end{tabular}
\end{equation}
\end{lemma}

By Equation~\eqref{eq:picbrequiv}, the homotopy groups of $\sPic^{\br}(\tSVect)$ are the homotopy groups of $\cZ(\mathbf{3SVect})^\times$. In order to obtain the homotopy groups in degree $\geq 1$ we note that $\Omega \cZ(\mathbf{3SVect})\cong \cZ_{(2)}(\tSVect)\cong \tSVect$, where the first equivalence is in \cite[Section IV.B]{JF} and the second holds by inspection. The homotopy groups of $\tSVect^\times$ are exactly the ones shown in degrees $\geq 1$.
By nondegeneracy of $\cZ(\mathbf{3SVect})$, we find that $\pi_0(\cZ(\mathbf{3SVect})^\times)\cong \Z/2$, by a higher $S$-matrix argument. The idea is that there are two components, and the 3d object represented by one of the components is detected by a line operator in the other component. 

We now consider maps $B^2 E \rightarrow B^2 \sPic^{\br}(\tSVect)$.  We break up the data of these maps as follows: first, we consider maps $\kappa:B^2 E \rightarrow \pi_2( B^2 \sPic^{\br}(\tSVect))$, represented by a class in $\rH^2(B^2 E, \Z_2)$. This twists the maps  $B^2E \rightarrow \pi_{\geq 3}( B^2 \sPic^{\br}(\tSVect))\simeq \SH^5$. Secondly, we have a map $B^2 E \rightarrow \SH^5$ compatible with $\kappa$. This is precisely the data of a class in $\SH^{5+\kappa}(B^2 E)$. The following result encapsulates the data of Items \ref{item:B1} and \ref{item:B2} in Theorem \ref{thm:mainB}:
\begin{corollary}
\label{cor:BFSF2C}
Braided fermionic strongly fusion 2-categories are  classified by a finite abelian group $E$, a class $\kappa \in \rH^2(B^2 E;\Z_2)$, and a class $\varsigma  \in \SH^{5+\kappa}(B^2 E)$.
\end{corollary}

In the case when $E = \Z_2$ and $\kappa$ is trivial, we have that $\SH^{5}(B^2 E) = 0$. The corresponding braided fermionic strongly fusion 2-category is $\mathbf{2SVect}_{\mathbb{Z}/2}$.
On the other hand, when $\kappa$ is the nontrivial class in $\rH^2(B^2 \mathbb{Z}_2;\Z_2)\cong\mathbb{Z}/2$, we have $\SH^{5+\kappa}(B^2 \Z_2)\cong\mathbb{Z}_2$. These two classes correspond to the nondegenerate braided fermionic fusion 2-categories denoted $\mathcal S = \mathcal{Z}(\mathbf{2SVect})$ and $\mathcal T$ in \cite{JF2,JFR}, the latter of which is not a Drinfeld center but has an anomaly given by $\int w_2 w_3$. In Sec. \ref{subsec:singlefermion} the (3+1)D fermionic toric code is equivalent to $\mathcal{S}$, and the decoherence gives the imTO $\tSVect$, with the 2-form symmetry generated by the fermion in $\Omega \tSVect\cong \SVect$.

Let $\fB=\tSVect^{(\kappa,\varsigma)}_{E}$ denote the braided fermionic strongly fusion 2-category corresponding to the data in the corollary above. Similar to the all bosons case in Section \ref{subsection:allboson}, consider an
 $H$-action on $\fB$  given by a functor $B H\to B \sAut^{\br}(\fB)$, where $G\twoheadrightarrow H$ and $N$ is the kernel of the surjection.  Let us first consider the case when  $H$ is trivial so that $N=G$, and we are solely classifying faithfully $N$-graded $N$-crossed braided extensions of $\fB$. Such categories are parametrized by homotopy classes of maps  $B N \to B\sPic(\tSVect^{(\kappa,\varsigma)}_{E})$ and understanding these maps will unpack in a more explicit manner the data of Item \ref{item:BmaptoPic} in Theorem \ref{thm:mainB}.

We begin by observing that the space $B\sPic(\tSVect^{(\kappa,\varsigma)}_{E})$ fits in the following fiber sequence
\begin{equation}\label{eq:pic2svectpullback}
B\sPic(\tSVect^{(\kappa,\varsigma)}_{E})\rightarrow B^2E\xrightarrow{(\kappa,\varsigma)}B^2\sPic^{\br}(\tSVect).
\end{equation}

   We consider first the group  $\pi_1(B\sPic(\tSVect^{(\kappa,\varsigma)}_{E}))$, which parametrize invertible $\tSVect^{(\kappa,\varsigma)}_{E}$-module 2-categories. By \cite[Lemma 3.17]{Decoppet:2025eic}, this group is trivial, hence we only have to look at maps $B N \rightarrow B \sPic(\tSVect^{(\kappa,\varsigma)}_{E})$ that land in homotopy groups greater than 1. By using the fiber sequence in Equation \eqref{eq:pic2svectpullback} we see that homotopy classes of maps $B N\rightarrow B\sPic(\tSVect^{(\kappa,\varsigma)}_{E})$ correspond to homotopy classes of maps $\tau: B N \rightarrow B^2 E$ together with a trivialization of the composite $\varsigma \circ \tau$, i.e. giving the lift in the following diagram:
\begin{equation}
\begin{tikzcd}
    B\sPic(\tSVect^{(\kappa,\varsigma)}_{E}) \arrow[r]& B^2 E \arrow[r,"\varsigma"]& B^2\sPic^{\br}(\tSVect) \\
    & B N\,. \arrow[u,"\tau"] \arrow[lu,dotted]&
\end{tikzcd}
\end{equation}

Such trivializations form a torsor over the group $\SH^{4+\kappa\circ\tau}(B N)$. In summary, Item \ref{item:BmaptoPic} amounts to the following pair of data:
    \begin{enumerate}
        \item A class $\tau\in \rH^2(B N;E)$, such that $\varsigma \circ \tau$ is trivial in $\SH^{5+\kappa\circ\tau}(B N)$,
        \item A class $\varpi \in \SH^{4+\kappa\circ\tau}(B N)$.
    \end{enumerate}

  We now consider when $H$ is nontrivial, so that there is a potentially anomalous weak symmetry. Then the  final two pieces of data in Theorem \ref{thm:mainB} has an interpretation analogous to what was given in Section \ref{subsection:allboson}.

\begin{remark}
    The theory $\tSVect^{(\kappa,\varsigma)}_E$ 
    parallels the theory  $\tVect^\beta_E$ that was used to construct mTOs. More generally, for any group $E$ and $\kappa\neq 0$, $\SH^{5+\kappa}(B^2 E)\cong \Z_2$, and the contribution from this group lies completely in the row of the spectral sequence with $\mathbb{C}^\times$ coefficients. In contrast to Sec.~\ref{subsection:allboson}, pure-state TOs with an emergent fermion do \emph{not} arise from trivial $E$: $\mathbf{2SVect}$ itself is degenerate, as the fermion is transparent. Rather, they arise from equivariantizing faithfully $G$-graded $G$-crossed braided extensions of $\mathcal{S}=\mathbf{2SVect}^{(\kappa,0)}_{\mathbb{Z}_2}$~\cite{Kong:2024ykr}, while extensions of $\mathcal{T}$ yield imTOs.
\end{remark}

\begin{remark}
    Unlike the description with emergent fermion excitations, when the topological order has local fermions then it is genuinely a fermionic theory and requires a spin structure rather than a weaker twisted spin structure. Categorically, these theories are realized as $\tSVect$-central braided fusion 2-categories i.e. a braided fusion 2-category $\fB$ with a sylleptic monoidal 2-functor  $\tSVect \to \cZ_{(2)}(\fB)$.
    We can realize these categories also from equivariantization of a $G$-action. Starting with a $G$-crossed strongly braided fusion 2-category $\fD = \boxplus_{g\in G}\fD_g$, we additionally require that there is a $G$-equivariant sylleptic monoidal 2-functor $\tSVect \to \cZ_{(2)}(\fD_e)$, where $\tSVect$ carries the trivial $G$-action.
\end{remark}

\subsection{Examples}
\label{subsection:egs}

We now supply a few examples of how the data of the parametrization of braided fusion 2-categories can be extracted from mTOs. Consider the theory in Sec.~\ref{subsec:nonabelianfloop}. The data in Theorem~\ref{thm:mainA} corresponds to:
\begin{itemize}
    \item The finite group $E = \Z_2 \times \Z_2$;
    \item The class $\beta \in\rH^5(B^2 E; \mathbb{C}^\times)$ is given by $\beta = \pi (B Sq^1 B + B' Sq^1 B')$ using the 2-form $\Z_2$ backgrounds $B, B'$.
    \item The group $N$ is trivial. The group $G=H = \Z_2$, acts on $\fC = \tVect^\beta_E$ by swapping the two $\Z_2$ generators. Namely, it acts on the $\Z_2$ backgrounds by swapping $B\leftrightarrow B'$.
\end{itemize}

We now explain how to produce the mTO defined via the category of 2-representations of a 2-group. While this is harder to realize from the point of view of a lattice model, the classification data allows us to pick out this theory. For the data of Theorem \ref{thm:mainA} we take $N$ to be trivial, so that $G=H$, but $E$ to be nontrivial and with $\beta$ trivial. We consider the action $\rho: B H \rightarrow  \mathrm{B}\mathscr{A}ut^{\br}(\mathfrak{B})$.  The homotopy groups of $\sAut^{\br}(\fB)$ for $\fB = \tVect_E$ are given by:
\begin{equation}\label{eq:groupsAutB}
    \begin{tabular}{|c|c|c|c|c|}
\hline
$\pi_0$ & $\pi_1$ & $\pi_2$ & $\pi_3$\\
\hline 
$\mathrm{Aut}(E)\ltimes \rH^4(B^2 E;\mathbb{C}^\times)$ & $0$ &  $\rH^2(B^2 E;\mathbb{C}^\times)$ & $0$ \\
\hline 
\end{tabular}\vspace{4mm}\,\,.
\end{equation}
The second entry is explained by $\pi_1(\sAut^{\br}(\fB))\cong\rH^3(B^2 E; \mathbb{C}^\times) = 0$.
An action $\rho: B H \rightarrow  B\mathscr{A}ut^{\br}(\mathfrak{B})$ therefore includes the data of a group homomorphism $H\rightarrow \mathrm{Aut}(E)\ltimes \rH^4(B^2 E;\mathbb{C}^\times)$. Assuming that this homomorphism factors through $\mathrm{Aut}(E)$, the only remaining data in $\rho$ is a class $\Theta \in \rH^3(B G;\underline{\widehat{E}})$, parameterizing maps from $B G$ to  $\pi_2(\sAut^{\br}(\fB))= \rH^2(B^2 E; \mathbb{C}^\times)\cong\widehat{E}$. 
We use the underline to denote that the fact that the
 coefficients are twisted by the homomorphism $G\rightarrow \mathrm{Aut}(E)$.  
The class $\Theta$ can be thought of as the obstruction to symmetry fractionalization of $G$ on $\widehat{E}$. Taking the equivariantization of the corresponding $G$-crossed braided fusion 2-category built from $\rho$ produces the braided fusion 2-category $\mathbf{2Rep}(\mathcal{G})$ of 2-representations of the 2-group $\mathcal G$ classified by $G$, $\widehat{E}$, and the class $\Theta$\footnote{One can also consider a trivial class $\Theta$ used to construct a split 2-group. A pure state TO is constructed from the data of a split 2-group is given in \cite{Gu:2026trc}.}.

\section{Conclusions and Future Directions}
\label{sec:discussion}

In this paper, we have developed a general framework for constructing and characterizing mixed-state topological order in (3+1)D. Motivated by lattice constructions, we have proposed an ansatz under which the topological operators of a (3+1)D mTO organize into a braided fusion 2-category, generalizing the proposed correspondence between (2+1)D mTO and non-modular braided tensor categories. Within this ansatz, we have provided a general recipe, built from an Abelian gauge theory subjected to decoherence and followed by weak/strong 0-form gauging, whose output exhausts all bosonic braided fusion 2-categories (including those with emergent fermions), and we have implemented it explicitly on the lattice for representative examples. Every such category hence admits a mixed-state realization along the lines we have described. We have further identified when these categories describe intrinsically mixed-state order: either the transparent operators carry a nontrivial anomaly, or the category is non-degenerate but carries the gravitational anomaly $\pi\int w_2w_3$, as for the FcFl theory (crucially, degeneracy alone does not suffice). Neither is expected in the ground state of a gapped local Hamiltonian in three spatial dimensions. This gives a proposed categorical criterion for distinguishing imTO from mTO in (3+1)D.

Our work raises a number of open questions and future directions. The most pressing open question is the extent to which the data of Theorems~\ref{thm:mainA} and~\ref{thm:mainB} furnish a classification of (3+1)D mixed-state topological order, rather than a means of construction. Our central ansatz asserts that the operators of an mTO organize into a braided fusion 2-category, but we have not shown that every operator algebra realizing a given category must arise this way, nor have we addressed when two sets of data $(E, \beta, N, H, \rho)$ (or their fermionic counterparts $(E,\kappa,\varsigma, N, H, \rho)$) describe the same mixed-state phase i.e., are related by finite-depth local quantum channels acting in both directions. Settling this would require fully understanding the equivalence relation that two-way channel connectivity induces on the categorical data, alongside a complete classification of braided fusion 2-categories in their own right. 

One can also ask a more restricted question. Rather than classifying all mTOs consistent with Ansatz~\ref{ansatz:mTO}, one could aim to classify only those that arise from decohering a (3+1)D topological order. Here, one starts from a pure-state topological order and views local decoherence as the process that can remove remote detectability. Physically, this amounts to classifying pairs of a parent pure state and a noise channel up to a suitable notion of equivalence; categorically, it should correspond to classifying the braided fusion 2-categories of Theorems~\ref{thm:mainA} and~\ref{thm:mainB} up to braided and monoidal equivalence, presumably a more tractable problem than the unconstrained classification discussed above. Relatedly, while our construction produces all braided fusion 2-categories at the level of categorical data, explicit lattice realizations for general data (for instance, arbitrary loop anomalies $\beta\in 
\rH^5(B^2E;\mathbb{C}^\times)$, or weak gauging with anomalous actions) remain to be constructed.

More broadly, we can relax the conditions on the ground states from which our construction starts. In (2+1)D, gapped ground states of local Hamiltonians are expected to realize only ordinary topological order; in (3+1)D, by contrast, gapped ground states can also realize fracton order, which lies outside the conventional TQFT paradigm altogether. Decohering a fracton phase can likewise generate new forms of mixed-state order: since braiding non-degeneracy is expected to hold even for fracton phases, mixed-state fracton orders that violate this principle will encode intrinsically mixed forms of fracton order. Formalizing this first requires an algebraic theory of braiding and fusion for fracton phases, on which there has been partial recent progress~\cite{song2019twist,song2024self,wickenden2026}. Separately, rather than obtaining mTOs by decohering a pure ground state, one may ask whether (3+1)D (i) mTOs can instead arise as steady states of open-system dynamics, extending recent work on steady-state SPT and topological order~\cite{chirame2024spt,chirame2024to,balasubramanian2024local}.

A further natural question concerns chirality. As recently emphasized in the $(2+1)$D setting~\cite{sun2026chiral,ellison2026chiral}, mixed states that appear chiral under certain definitions of chirality can fail to be chiral under other definitions, and the usual diagnostics of chirality for pure-state topological order do not carry over unmodified to mTOs. This question is sharpened in (3+1)D by the fact that, while (3+1)D topological orders always admit a gapped boundary, this need not be true for a general mTO: since a generic braided fusion 2-category need not arise as the center of any fusion 2-category, it need not admit a gapped boundary at all. Understanding chirality for mTOs in a precise physical sense, and, relatedly, classifying which mTOs admit boundaries at all (gapped or otherwise) and what data characterizes them, will require ideas beyond those available for pure states.

Finally, it is worth revisiting notions of orientability for mixed-state theories. For pure-state (3+1)D topological order, framed theories correspond to nondegenerate braided fusion 2-categories, while oriented theories require the additional data trivializing the associated anomaly (the categorical analogue of the distinction between nondegenerate braided fusion categories and modular tensor categories in (2+1)D). Since decoherence deforms what is a priori a framed pure-state theory, it is natural to ask whether the resulting mTOs should likewise be regarded as framed, and if so, what plays the role of a framing anomaly in this setting. Formalizing the notion of a mixed-state TQFT along the lines of Ref.~\cite{ziniwang} and using it to settle this question would help place mTOs on the same functorial footing as their pure-state counterparts.

\begin{acknowledgments}
We thank Thibault Décoppet for helpful discussions. RK is supported by the Department of Applied Physics, the University of Tokyo. AP is supported by the U.S. Department of Energy, Office of Science, Office of Advanced Scientific Computing Research, through the Exploratory Research for Extreme Scale Science (EXPRESS) program under Award No. DE-SC0026216. MY is supported by the EPSRC Open Fellowship EP/X01276X/1, and the Fund for Natural Science and the Adler Family Fund. We thank the Galileo Galilei Institute for Theoretical Physics for hosting the program “Defects and Extended Excitations in Quantum Field Theory, Quantum Matter and Statistical Models” in 2026, during which part of the work was completed. MY would like to thank the Institute for Advanced Study for hosting him during a visit where he learned about mixed-state topological order.
\end{acknowledgments}


\bibliography{bibliography}

\appendix

\setcounter{theorem}{0}
\setcounter{defn}{0}
\renewcommand{\thetheorem}{\Alph{section}\arabic{theorem}}
\renewcommand{\thedefn}{\Alph{section}\arabic{defn}}
\makeatletter
\def\theHtheorem{app.\Alph{section}.\arabic{theorem}}
\def\theHdefn{app.\Alph{section}.\arabic{defn}}
\makeatother


\section{Topological Order and Fusion 2-Categories}
\label{appendix:F2C}

In order to describe the operators of a  (3+1)D mTO, or pure-state topological order, we use the language of fusion 2-categories. A fusion 2-category categorifies a fusion 1-category, the latter which is relevant in the study of  symmetries and topological order in (2+1)D TOs. In recent years, the study of fusion 2-categories has led to advancements of our understanding of many important areas including: the study of finite symmetries in (2+1)D, (3+1)D and (4+1)D topological order \cite{PhysRevX.8.021074,JFR,JFY2,Decoppet:2022dnz,Debray:2023iwf,Bhardwaj:2024qiv,Wen:2024qsg,3TY,Bhardwaj:2025piv,Schafer-Nameki:2025fiy,Wen:2025thg,TY25}, anomalous (3+1)D symmetry enriched topological order, non-invertible duality symmetry in (3+1)D, and how to prove the existence of a minimal non-degenerate extension for braided fusion 1-categories. Furthermore, their classification has been given in \cite{F2C}, with the main result revealing that the purely 2-categorical part is easier to deal with than one might think: it just reduces to group theoretical and cohomological data. In this Appendix, we introduce the mathematical formalisms and background regarding both bosonic and emergent fermionic braided fusion 2-categories. For the purposes of conveying physical intuition, we start with only bosonic fusion 2-categories and we relegate the discussion of how the cohomological data changes when there is an emergent fermion in the spectrum for Sec~\ref{subsection:fermionicbackground}.

\subsection{Details on Fusion Categories}
\label{subsection:defs}
We give a quick account of the details and terminology associated to  fusion $n$-categories (over $\mathbb{C}$). For the reader that wants an account of how to rigorously define fusion $n$-categories in general we recommend \cite[Section 3.1]{3TY}, and \cite{GJF,Kong:2024ykr,Stockall:2025rng}.
\begin{itemize}
\item An $n$-category $\mathbf{C}$ is called \textit{locally finite semisimple} if for every object $X \in \mathbf{C}$, $\Omega_X \mathbf{C} := \End_{\mathbf{C}}(X)$ is a multifusion $(n-1)$-category. In particular, every $k$-morphism for $0<k<n$ has both a left and a right adjoint. We will denote the looping at the identity object in $\mathbf{C}$ by $\Omega \mathbf{C}$.

\item An $n$-category $\mathbf{C}$ is called \textbf{semisimple} if it is locally finite semisimple, and it is Cauchy complete in the sense of \cite{GJF}, that is, every $n$-condensation monad splits and direct sums for objects exist.

Cauchy completion is the statement of condensation completion:\ i.e.\ we include all possible $n$-dimensional topological defects that can be constructed by condensing/gauging topological defects living inside an $n$-dimensional topological defect $X$, which forms the multifusion $(n-1)$ category $\Omega_X \mathbf{C}$.

\item An object $X$ of a semisimple $n$-category $\mathbf{C}$ is called \textit{simple} if $\Omega_X \mathbf{C}$ is fusion. 

This means that there are no genuine topological point-like defects on $X$ except for multiples of the identity local operator.
\item Two simple objects are called \textit{connected} if there exists a non-zero 1-morphism between them. If $\mathbf{C}$ is locally finite, being connected defines an equivalence relation on the set of equivalence classes of simple objects of $\mathbf{C}$, and we write $\pi_0(\mathbf{C})$ for the corresponding quotient.

Physically, two $n$-dimensional topological defects are connected if there exists a non-zero topological interface between them. Two topological defects that are connected are also related by condensations, or in other words either of the two defects can be obtained by gauging topological sub-defects living on the other defect. Thus, $\pi_0(\mathbf{C})$ characterizes gauge inequivalent classes of topological $n$-dimensional defects. We note that having multiple objects which live in the  same component is not a feature of fusion 1-categories. This is readily seen, as there are no nontrivial morphisms between anyons of different type in a (2+1)D TO. The next definition categorifies the notion of finiteness.

\item A \textit{finite semisimple $n$-category} is a semisimple $n$-category such that $\pi_0(\mathbf{C})$ is finite. For $n=1$, connected components are exactly (equivalence classes of) simple objects. For $n=2$, every connected component contains only finitely many simple objects~\cite{DR}. 

\item A \textit{multifusion} $n$-category $\mathbf{C}$ is a monoidal finite semisimple $n$-category that is rigid, i.e.\ every object has a left and a right dual.

Physically, such an $n$-category describes topological defects that can be bent into $\cup$ and $\cap$ shapes.
\item A \textit{fusion} $n$-category is a multifusion $n$-category whose monoidal unit is simple.
\end{itemize}

Fusion $n$-categories are useful for describing and classifying topological orders, up to gapped boundary,  due to the axiomatization that topological order in $(n+1)$D are fusion $n$-categories with trivial center, enforcing remote detectability. In terms of braided categories, one can say that $(n+1)$D TOs are given by nondegenerate braided fusion $(n-1)$-categories. 
In (3+1)D, the constrained structure of braided fusion 2-categories makes possible a systematic classification of (m)TOs, whereas the analogous question in (2+1)D seems out of reach. 

\subsection{Braided Fusion 2-Categories}\label{subsection:BF2C}

We now define braided fusion 2-categories, and then specify what we mean by nondegenerate braided fusion 2-category.

\begin{defn}\label{def:BF2C}
    Suppose $\fB$ is a monoidal 2-category with tensor bifunctor $\otimes$, associator $\alpha$, and pentagonator $\pi$. A \textit{braiding} on $\fB$ consists of a natural (in both variables) equivalence 
\begin{equation}\label{eq:braidingF2C}
    b_{x|y}:x \otimes y \to y \otimes x,
\end{equation}
 together with \textit{hexagonators} $R_{(x|-,-)}$ and $S_{(-,-|x)}$ that provide the monoidality of $b$:
 \begin{widetext}
\begin{equation}
\adjustbox{scale=0.75}{
\begin{tikzcd}[column sep = tiny]
                                                                    & (y \otimes x) \otimes z \arrow[r, "\alpha "]     & y \otimes (x \otimes z) \arrow[rd, "{b_{x|z}}"]       &                         \\
(x\otimes y)\otimes z \arrow[rd, "\alpha "] \arrow[ru, "{b_{x|y}}"] &       \hspace{15mm} \Downarrow                   & {\hspace{-15mm}R_{(x|y,z)}}                           & y \otimes( z \otimes x) \\
                                                                    & x \otimes (y \otimes z)  \arrow[r, "{b_{x|y z}}"] & (y \otimes z )\otimes x \arrow[l, phantom] \arrow[ru,"\alpha"] &                        
\end{tikzcd}
},
\quad
\adjustbox{scale=0.75}{
\begin{tikzcd}[column sep = tiny]
                                                                    & (z \otimes y) \otimes x \arrow[r, "{b_{zy|x}} "]     & x \otimes (z \otimes y) \arrow[rd, "\alpha"]       &                         \\
z\otimes (y\otimes x) \arrow[rd, "{b_{y|x}} "] \arrow[ru, "\alpha"] &       \hspace{15mm} \Downarrow                   & {\hspace{-15mm}S_{(z,y|x)}}                           & (x \otimes z) \otimes y \\
                                                                    & z \otimes (x \otimes y)  \arrow[r, "\alpha"] & (z \otimes x )\otimes y\,. \arrow[l, phantom] \arrow[ru,"{b_{z|x}}"] &                        
\end{tikzcd}}
\end{equation}
$R$ and $S$ must solve various equations in \cite[Appendix C.2]{SP}.
\end{widetext}
\end{defn}

\begin{ex}
In Definition \ref{def:BF2C} the data of $\alpha$ and $b$ are automatically trivial
in the bosonic strongly fusion case.
However, there is still data of a 4-cocycle $\pi \in Z^4(B G; \mathbb{C}^\times)$ that one may  use to twist the pentagonator $\tVect_G$. This produces the fusion 2-category $\tVect^\pi_G$. The resulting fusion 2-category depends up to monoidal equivalence only on the class of $\pi$ in the cohomology group $\rH^4(B G;\mathbb{C}^\times)$.  By \cite[Theorem A]{JFY} for $G$ a finite group all bosonic strongly fusion 2-categories arise as $G$-graded extensions of $\tVect$, and take the form $\tVect^\pi_G$.
\end{ex}

 \begin{defn}
     The triple $(\pi,R,S)$ consisting of pentagonator $\pi$ and hexagonators $R$ and $S$, is the data of a \textit{braided cocycle}: they define a class in the Eilenberg--Mac Lane cohomology
\begin{equation}\label{eq:cohobraided}
 [\pi,R,S] \in \rH^4_{\mathrm{br}}(E;\mathbb{C}^\times) := \rH^5(B^2 E;\mathbb{C}^\times),
\end{equation}
for an abelian group $E$. 
 \end{defn}

For a given braided fusion 2-category, in order to define what we mean by nondegenerate, and for the purposes of our subsequent discussion, we introduce the sylleptic center for braided monoidal 2-categories, which categorifies the symmetric center of a braided monoidal
1-category.

\begin{defn}[\cite{Cr}]\label{def:syllepticcenter}
    The sylleptic center of a braided fusion 2-category $\fB$, denoted by $\cZ_{(2)}(\fB)$, is the sylleptic fusion 2-category defined as follows
    \begin{itemize}
        \item Objects given by a pair $(A,v_{A,-})$ where $A$ is an object in $\fB$ and $v_{A,-}$ is an invertible modification:
        \begin{equation}
            \begin{tikzcd}[sep=tiny]
                A\otimes X \arrow[rr,"\cong"] \arrow[rdd,"b_{A|X}",swap]&  &A\otimes X\\
               & {\Downarrow v_{A||X}} & \\
                & X\otimes A \arrow[ruu,"b_{X|A}",swap]&
            \end{tikzcd}
        \end{equation}
        \item The 1-morphisms between $(A,v_A)$ and $(A',v_{A'})$ are 1-morphisms $f:A \to A'$ in $\mathfrak{B}$ such that the following diagram commutes:
         \begin{equation}
            \begin{tikzcd}[sep=tiny]
                A\otimes X \arrow[rr,"\cong"] \arrow[ddd]\arrow[rdd,"b_{A|X}",swap]&  &A\otimes X \arrow[ddd]\\
               & {\Downarrow v_{A||X}} & \\
                & X\otimes A \arrow[ruu,"b_{X|A}",swap] \arrow[ddd]& \\
                A'\otimes X \arrow[rr] \arrow[rdd,"b_{A'|X}",swap]& & A'\otimes X \\
                 & {\phantom{aaa}\Downarrow v_{A'||X}} & \\
                 & X\otimes A' \arrow[ruu,"b_{X|A'}",swap]&
            \end{tikzcd}
        \end{equation}
        \item The 2-morphisms are exactly the 2-morphisms in $\mathfrak{B}$.
    \end{itemize}
\end{defn}

We say that a braided fusion 2-category $\fB$ is nondegenerate if $\cZ_{(2)}(\fB)\cong \tVect$. In the case when the (3+1)D TO has only bosonic excitations, the theory corresponds to a nondegenerate braided fusion 2-category and the classification of nondegenerate braided fusion 2-categories can be conducted through a 2-categorical version of (de-)equivariantization. We will give further details on this classification in \S\ref{section:deequiv}.

\subsubsection{$G$-Crossed Braided and Braided Extensions for Fusion 2-Categories}\label{subsection:enrichment}
We start with a quick physical recount in (2+1)-dimensions of coupling a topological order to an ordinary ($0$-form) symmetry. A (2+1)d topological order described by a nondegenerate braided fusion category $\mathcal B$ can be coupled to a symmetry group $G$ by introducing topological symmetry defects labeled by elements $g \in G$. Transporting a quasiparticle, represented by an object of $\mathcal B$, around a $g$-defect implements the action of $g$ on its topological charge, which may permute quasiparticle types. This structure is captured by a $G$-crossed braided fusion category
\[
\mathcal A = \bigoplus_{g \in G} \mathcal A_g,
\]
whose identity component is the original category $\mathcal A_e = \mathcal B$. The category $\mathcal A$ is therefore a $G$-crossed braided extension of $\mathcal B$.
From a physical perspective, each sector $\mathcal A_g$ describes a distinct type of topological $g$-defect, and the fusion and braiding properties of these defects, as well as their interactions with ordinary quasiparticles, are governed by the $G$-crossed braided structure.

 One can present a homotopy theoretic classification of faithfully graded $G$-crossed braided fusion 2-categories, similar to the classification of faithfully graded $G$-crossed braided fusion 1-categories obtained in \cite{ENO2}. In order to do so, we introduce the Picard space $\sPic(\fB)$ of a braided fusion 2-category $\mathfrak{B}$. By definition, the space $$\sPic(\fB):=\mathbf{Mod}(\mathfrak{B})^{\times}$$ consists of invertible $\fB$-module 2-categories and invertible $\fB$-module morphisms. 
 
 The homotopy groups of $\sPic(\fB)$ have the following description:
 \begin{equation}
\begin{tabular}{|c|c|c|c|c|}
\hline
$\pi_0$ & $\pi_1$ & $\pi_2$ & $\pi_3$ \\
\hline 
$Pic(\fB)$ & $Inv(\fB)$ & $Inv(\Omega\fB)$ &  $\mathbb{C}^{\times}$ \\
\hline
\end{tabular}
\end{equation}
 where $Inv(-)$ denotes the invertible elements.
 This space is a group-like topological monoid under the relative tensor product, and may therefore be delooped to the space $B \sPic(\fB)$.  Just as in the case of braided fusion 1-categories, homotopy classes of maps from $B G \rightarrow B \sPic(\fB)$ parametrize faithfully $G$-crossed braided extensions of $\fB$ by the following theorem, which generalizes the results of \cite{ENO2}: 
 
\begin{theorem}\label{thm:crossedextn}
Let $\fB$ be a braided fusion 2-category. Faithfully graded $G$-crossed braided extensions of $\fB$ are classified by homotopy classes of maps $B G \rightarrow B \sPic(\fB)$.
\end{theorem} 

Understanding the maps in Theorem \ref{thm:crossedextn} is indispensable for the classification of pure-state TOs.
We explain what we mean by a faithfully graded $G$-crossed braided extension in three steps, beginning with faithful $G$-graded extensions. Given a fusion 2-category $\mathfrak{C}$, faithfully $G$-graded extensions $\mathfrak{D}=\boxplus_{g\in G}\mathfrak{D}_g$ of $\mathfrak{C}=\mathfrak{D}_e$ are classified by homotopy classes of maps $B G\rightarrow B \sBrPic(\mathfrak{C})$, where $\sBrPic(\mathfrak{C})$  is the space of invertible $\mathfrak{C}$-$\mathfrak{C}$-bimodule 2-categories. We now upgrade from a faithfully $G$-graded extension to a faithfully graded $G$-crossed extension. To do this we supply the data of a $G$-action on $\mathfrak{C}$, and hence an action of $G$ on $\sAut^{br}(\cZ(\mathfrak{C}))$. We endow the composition 
\begin{equation}
    B G\rightarrow B \sBrPic(\mathfrak{C})\rightarrow B\sAut^{\mathrm{br}}(\mathcal{Z}(\mathfrak{C}))
\end{equation}
with a $G$-equivariant structure, with $G$ acting by conjugation. In doing so, this corresponds to
equipping $\mathfrak{D}$ with $G$-action that acts by conjugation on the grading. In the last step, we explain how to upgrade to a faithfully graded $G$-crossed braided extension. Suppose that $\mathfrak{D}=\boxplus_{g\in G}\mathfrak{D}_g$ where $\mathfrak{D}_e = \fB$ where $\fB$ is a braided fusion 2-category. Furthermore, suppose that the $G$-action on $\fD$ is compatible with the braiding on $\fB$. The data of a $G$-action on $\fB$ is given by a map $B G \rightarrow \sAut^{\br}(\fB)$,  which also composes to an action on $\sAut^{\br}(\cZ(\fB))$.
We observe that the data of a $G$-crossed braiding on $\mathfrak{D}$ that is compatible with the braiding and $G$-action on $\mathfrak{B}=\mathfrak{D}_e$ correspond precisely to the data of a homotopy witnessing the commutativity of the following diagram of spaces:
\begin{equation}
\begin{tikzcd}[sep=small]
{\mathrm{B}G} & {\mathrm{B}\mathscr{A}ut^{\mathrm{br}}(\mathfrak{B})}\\
{\mathrm{B}\mathscr{B}r\mathscr{P}ic(\mathfrak{B})} & {\mathrm{B}\mathscr{A}ut^{\mathrm{br}}(\mathcal{Z}(\mathfrak{B})).}
\arrow[from=1-1, to=1-2]
\arrow[from=1-1, to=2-1]
\arrow[Rightarrow, from=2-1, to=1-2, shorten >= 2ex, shorten <= 2ex]
\arrow[from=1-2, to=2-2]
\arrow[from=2-1, to=2-2]
\end{tikzcd}
\end{equation}

By noting that there is a canonical pullback of spaces given by:
\begin{equation}\label{eq:DavydovNikshychpullback}
\begin{tikzcd}[sep=small]
{\mathrm{B}\mathscr{P}ic(\mathfrak{B})} & {\mathrm{B}\mathscr{A}ut^{\mathrm{br}}(\mathfrak{B})}\\
{\mathrm{B}\mathscr{B}r\mathscr{P}ic(\mathfrak{B})} & {\mathrm{B}\mathscr{A}ut^{\mathrm{br}}(\mathcal{Z}(\mathfrak{B})),}
\arrow[from=1-1, to=1-2]
\arrow[from=1-1, to=2-1]
\arrow["\lrcorner"{anchor=center, pos=0.02}, draw=none, from=1-1, to=2-2]
\arrow[from=1-2, to=2-2]
\arrow[from=2-1, to=2-2]
\end{tikzcd}
\end{equation}
we see that homotopy classes of maps $B G \rightarrow B \sPic(\fB)$ give rise to faithfully graded $G$-crossed braided extensions of $\fB$.

In \cite[Section 4]{davydov2021braided}, the authors develop an alternative extension theory specifically adapted to braided fusion 1-categories. More concretely, they classify braided extensions of braided fusion  1-categories in terms of the braided Picard space. In this work, we adopt a 2-categorical analogue of this approach, which likewise makes essential use of the \textit{braided Picard space} (generalizing \cite[Section 5.3]{davydov2021braided})

\begin{equation}
\label{eq:picbrequiv}
    \sPic^{\mathrm{br}}(\fB) := \cZ(\Mod (\fB))^\times. 
\end{equation}
This is a braided group-like topological monoid, so that it can be delooped twice. We now show how the space $B^2 \sPic^{\br}(\fB)$ can be used to classify braided extensions, i.e. construct new braided categories with symmetry.

\begin{prop}
\label{prop:braidedextn}
    Let $\fB$ be a braided fusion 2-category, and $E$ a finite abelian group. Faithfully $E$-graded braided extensions of $\mathfrak{B}$ are classified by homotopy classes of maps $B^2 E \rightarrow B^2 \sPic^{\mathrm{br}}(\mathfrak{B})$.
\end{prop}

The final gadget that we introduce is the 3-groupoid of braided autoequivalences of a braided fusion 2-category $\fB$, denoted $\sAut^{\br}(\fB)$. It is important to note that there is a map $\sPic(\fB) \rightarrow \sAut^{\br}(\fB)$ but it is not an equivalence on low degree homotopy groups\footnote{One can contrast this with the fact that if $\cC$ is a fusion 1-categories, then $\sPic(\cZ(\cC)) \rightarrow \sAut^{\br}(\cZ(\cC))$ is an equivalence.}. In particular, the group $\pi_0\sAut^{br}(\fB)$ is the group of braided autoequivalences $F: \fB \to \fB$ which preserve the braiding on objects, and is invertible up to braided monoidal natural isomorphism e.g. in (2+1)d, this captures all the possible permutation of anyons which preserve the modular data. The higher homotopy groups of $\sAut^{\br}(\fB)$ capture the coherent higher equivalences between braided autoequivalences.
Lastly, as $\sAut^{\br}(\fB)$ is a grouplike space, it too can be delooped.

\subsection{(De-)equivariantization and Nondegenerate Braided Fusion 2-Categories}\label{section:deequiv}
 The (de-)equivariantization procedure induces an equivalence between the 2-category of braided fusion 1-categories containing $\Rep(G)$ and the 2-category of $G$-crossed braided fusion 1-categories \cite{DGNO}. We require an analogous statement for fusion 2-categories, where we now must contend with a $\tRep(G)$ action. Because $\Rep(G)$ is symmetric monoidal, we can always take iterated modules, thus we can view $\Mod(\Rep(G)):= \tRep(G)$ as the condensation surfaces formed by the line operators of the theory. The equivalence that the 2-categorical (de-)equivariantization implements is due to the fact that a 0-form $G$-symmetry in (3+1)D is dual to a categorical 2-form symmetry, for which we condensation complete to obtain all the surface operators given by $\tRep(G)$\footnote{Strictly speaking, (de-)equivariantization differs from gauging in that it does not add twisted sectors. For a strong symmetry $N$, these are supplied by the faithfully $N$-graded extension $\mathcal{C}$, so that gauging $N$ amounts to extension followed by equivariantization. For a weak symmetry $H$, no deconfined twisted sectors arise and equivariantization alone implements classical gauging. We use ``gauging'' in this sense in the main text.}.

\begin{theorem}[{\cite[Theorem 3.2]{Decoppet:2025eic}}]
\label{thm:equiv}
    There is an equivalence of 3-categories between braided fusion 2-categories containing $\tRep(G)$ fully faithfully and $G$-crossed braided fusion 2-categories whose
trivially graded component is nondegenerate.
\end{theorem}

Let us explain the idea of this theorem. Suppose that $\fB$ is a braided fusion 2-category representing an imTO, and we equip it with a $G$-symmetry i.e. a $G$-grading. We would like to show that gauging $G$ on $\fB$ (equivariantizing by $G$ on $\fB$) produces a  nondegenerate category if and only if the $G$-grading is faithful and if $\fB_e$, the trivially graded component, is non-degenerate.   Let $N\trianglelefteq G$ be the normal subgroup of $G$ on which the grading of $\mathfrak{B}$ is supported. We consider this as the strong $N$-symmetry. Gauging $G$ produces a dual $\tRep(G)$ symmetry enjoyed by $\fB^G$. From the definition of a $G$-crossed braiding, a line operator in $\mathbf{Rep}(G)=\Omega\mathbf{2Rep}(G)$, can be lifted to $\mathcal{Z}_{(2)}(\mathfrak{B}^G)$ if and only if the subgroup $N$ acts trivially on it. This implies there is an embedding $\mathbf{Rep}(G/N)\hookrightarrow \Omega\mathcal{Z}_{(2)}(\mathfrak{B}^G)$. Let us now consider when the $G$ grading on $\mathfrak{B}$ is faithful, that is, $G=N$ is the entire strong symmetry. But, under this assumption, we have $\mathcal{Z}_{(2)}(\mathfrak{B}^G)\simeq \mathcal{Z}_{(2)}(\mathfrak{B}_e)$. Hence if $\fB_e$ is also nondegenerate, then we have proven the theorem.

As an immediate fact that can be derived from the above theorem is that there is an equivalence of 3-categories between nondegenerate braided fusion
2-categories containing $\tRep(G)$ and \textit{faithfully graded} $G$-crossed braided fusion 2-categories whose trivially graded component is nondegenerate. Hence, (de-)equivariantization preserves nondegeneracy, provided the symmetry acts non-anomalously. 

With these facts in place, we present the classification of nondegenerate bosonic braided fusion 2-categories. Such categories, which correspond to pure-state (3+1)D TOs where all the point-like excitations are bosonic, are classified up to gapped boundary by a finite group $G$ and a class $\pi \in \rH^4(B G; \mathbb{C}^\times)$~\cite{PhysRevX.8.021074}. In particular they are all realized as Drinfeld centers $\cZ(\tVect^\pi_G)$, i.e. finite $G$-gauge theories with Dijkgraaf-Witten action presented by $\pi$.

\subsection{Topological Order with Emergent Fermions}\label{subsection:fermionicbackground}
We now introduce the mathematical tools needed to classify mixed state (3+1)D TOs with emergent fermions. An essential tool that replaces the classification data used in the case where the spectrum contains only bosons is a generalized cohomology theory called \textit{supercohomology} \cite{Wang:2017moj}, which we denote by $\SH^*$. We hence review supercohomology and its twisted cousins which are necessary for explaining the data of mTOs with emergent fermions.

Denote the spectrum representing supercohomology by $\SH$, whose
homotopy groups are given by
\begin{equation}
    \pi_{-2}\SH = \Z_2 , \quad  \pi_{-1}\SH = \Z_2, \quad  \pi_{0}\SH = \mathbb C^\times\,,
\end{equation}
with Postnikov extension data given by
\begin{equation}\label{eq:postnikov}
    \Sq^2:\pi_{-2}\SH \rightarrow \pi_{-1}\SH, \quad (-1)^{\Sq^2}:\pi_{-1}\SH \rightarrow \pi_{0}\SH.
\end{equation}
The above homotopy groups suggest that supercohomology is inherently tied to the space of invertible fermionic gapped phases in spacetime dimensions 2 and below. Indeed, in dimension (0+1)D, the only nontrivial  invertible fermionic phase is the fermion, and in dimension (1+1)D it is the Majorana chain. There is another way to realize supercohomology that is purely  categorical. Namely, it is given as a shift of the homotopy groups of the 2-groupoid 
\begin{equation}\label{eq:2SVectgroups}
    (\tSVect)^\times = \Z_2\boldsymbol{\cdot}B\Z_2\boldsymbol{\cdot}B^2 \mathbb{C}^\times.
\end{equation}
Physically, the $B\Z_2$ represents the vacuum and the emergent fermion line $f$, and the $\Z_2$ represents the condensation of $f$ on a surface.
A natural question one might ask is why such a generalized cohomology theory is necessary for describing TQFTs in (3+1)D. The reason for this is because supercohomology is indispensable for understanding the structure of fusion 2-categories which support emergent fermions, see e.g. the classification of fusion 2-categories \cite{F2C,Dec24,TY25}, and hence for classifying TQFT in (3+1)D as well as constructing anomalous TQFT in (3+1)D \cite{Debray:2025kfg,Debray:2026pet}.

Supercohomology is  used in the parametrization of Theorem \ref{thm:mainB}, in its twisted form. This twisting arises because $(\tSVect)^\times$ has a nontrivial space of automorphisms given by 
$$\sAut^{\mathrm{br}}(\SVect) \simeq B \Z_2\,,$$ with nontrivial braided natural isomorphism given by the fermion parity operator $(-1)^F$. 
\begin{defn}
    Given a space $X$ equipped with a map $\kappa:X\to  B^2\Z_2$, or, equivalently, a $B\Z_2$-equivariant structure, the \textit{$\kappa$-twisted $n$-th supercohomology group} of $X$ is the group $\SH^{n+\kappa}(X)$ of homotopy classes of $B\Z_2$-equivariant maps from $X\rightarrow B^{n-2}\tSVect^\times$. In other words, maps over $B \Z_2$ of the form
\begin{equation}
    \begin{tikzcd}
        X \arrow[rr] \arrow[rd,"\kappa"]& & B^{n-2}\tSVect^\times\arrow[dl]\\
 &B^2 \Z_2\,. &  
    \end{tikzcd}
\end{equation}
\end{defn}

\end{document}